\documentclass[%
 reprint,
superscriptaddress,
 amsmath,amssymb,
 aps,
]{revtex4-2}
\pdfoutput=1
\usepackage{amsmath}
\usepackage{amssymb}
\usepackage{hyperref}
\usepackage{amsthm}
\usepackage{enumerate}
\usepackage{array}
\usepackage{breqn}
\usepackage{tabularx}
\usepackage{physics}
\usepackage{tikz}
\usepackage{tabularx}
\usepackage{multirow}
\usepackage{tikzit}
\tikzstyle{black dot}=[fill=black, draw=black, shape=circle]
\tikzstyle{white dot}=[fill=none, draw=black, shape=circle]
\tikzstyle{black edge}=[-]
\theoremstyle{plain}
\newtheorem{theorem}{Theorem}[section]

\newtheorem{lemma}[theorem]{Lemma}
\newtheorem{claim}[theorem]{Claim}
\newtheorem{proposition}[theorem]{Proposition}

\theoremstyle{definition}
\newtheorem{definition}[theorem]{Definition}
\newtheorem{example}[theorem]{Example}

\theoremstyle{remark}

\DeclareMathOperator{\nbhd}{N}

\newcommand{\tens}[1]{
  \mathbin{\mathop{\otimes}\limits_{#1}}
}

\begin{document}

\title{Improved Graph Formalism for Quantum Circuit Simulation}
\author{Alexander Tianlin Hu}
\affiliation{The~Harker~School}
\author{Andrey Boris Khesin}
\affiliation{MIT}
\date{\today}
\begin{abstract}
Improving the simulation of quantum circuits on classical computers is important for understanding quantum advantage and increasing development speed. In this paper, we explore a new way to express stabilizer states and further improve the speed of simulating stabilizer circuits with a current existing approach. First, we discover a unique and elegant canonical form for stabilizer states based on graph states to better represent stabilizer states and show how to efficiently simplify stabilizer states to canonical form. Second, we develop an improved algorithm for graph state stabilizer simulation and establish limitations on reducing the quadratic runtime of applying controlled-Pauli $Z$ gates. We do so by creating a simpler formula for combining two Pauli-related stabilizer states into one. Third, to better understand the linear dependence of stabilizer states, we characterize all linearly dependent triplets, revealing symmetries in the inner products. Using our novel controlled-Pauli $Z$ algorithm, we improve runtime for inner product computation from $O(n^3)$ to $O(nd^2)$ where $d$ is the maximum degree of the graph.
\end{abstract}
\maketitle
\section{Introduction}
In quantum computation and quantum information, the \textit{stabilizer formalism} is a way of working with a particular set of quantum states. The state vectors of these states are stabilized, namely belong to the eigenspace of Hermitian Pauli operators \cite{mikeike}. 
Working with these operators is much easier than working with the state vectors themselves, and linear combinations of stabilizer states are often used to represent quantum states compactly. The stabilizer formalism is first introduced to describe the codes and circuits used in quantum error correction and fault-tolerant quantum computation \cite{gottesman, gottesman phd}. 
Later, it is applied to simulate quantum circuits \cite{bravyi smolin}, where a universal class of quantum circuits is simulated by expressing magic states as linear combinations of stabilizer states \cite{bss,kerzner,bravyi smolin,chp,stim,hq,shir}.

\textit{Graph states} are a special type of stabilizer state that can be completely defined in terms of a graph. Graph states, like stabilizer states, have many applications, such as a resource in measurement based quantum computing \cite{cluster}, constructing codes in quantum error correction \cite{stabgraph3,mh2}, and representing stabilizer states in the classical simulation of stabilizer circuits \cite{graphsim}.

The graph and its properties give a nice structure to study graph states. For example, graph-theoretic properties can be leveraged to characterize and compute multi-party entanglement in graph states \cite{multiparty}. The orbits of graph states under local Clifford operations can be generated by applying local complementation to the graph \cite{van den nest}.

Graph states can be extended to represent all stabilizer states by applying local 
Clifford operators to each of the qubits \cite{van den nest}. 
In this paper, we study stabilizer states through the lens of graph states augmented with local Clifford operations.
This approach has advantages because representations of stabilizer states with graph states provide a way of visualizing stabilizer states.
The goal of our research is to connect graph-theoretic properties with properties of the stabilizer state and improve the classical simulation of quantum circuits. 
Our paper consists of Section \ref{section2}, containing key definitions, followed by Sections \ref{section3}, \ref{section5}, \ref{section4}, each containing one of our contributions.

In Section \ref{section3}, we discover and propose a new \textit{canonical form} for stabilizer states based on graph states, which improves upon previous work by being unique.
In \cite{stabgraph1,stabgraph2}, it is shown that any stabilizer state expressed with graph states can be further simplified to a simple and elegant \textit{reduced form}.
However, two stabilizer states in reduced form may appear different yet actually refer to the same state. 
To test whether two stabilizer states in reduced form refer to the same state, an equivalence test is developed in \cite{stabgraph1}.
Another canonical form for stabilizer states is developed in \cite{garcia}. It is based on the binary matrix representation. It can be used to compute inner products efficiently and does so by being converted to a canonical circuit with five blocks of Clifford gates.
We develop our canonical form by further simplifying the reduced form. Testing for equivalence is now a trivial comparison. The graph provides a nice structure to work with. The canonical form can directly be converted to a four-layer canonical circuit, and our line of work can be further studied to simplify diagrammatic representations of Clifford operators in the ZX calculus \cite{backen}.

In Section \ref{section5}, we develop a simpler and faster algorithm for graph state simulation, the method of using graph states to simulate stabilizer circuit computations.
Graph state simulation is suited for simulating circuits with fewer interactions between qubits and nearly independent pieces.
Currently, the first algorithm for graph state simulation was GraphSim, introduced and implemented in \cite{graphsim}. 
In \cite{stim}, it is shown that GraphSim is the fastest method for simulating multi-level S state distillation circuits, compared to the CHP simulator \cite{chp}, Qiskit simulator \cite{qiskit}, Cirq simulator \cite{circ}, and Stim simulator \cite{stim}.
GraphSim applies controlled-Pauli $Z$ gates to graph states augmented with local Clifford operations by applying local complementations to graphs iteratively with many cases depending on the local Clifford operators.
GraphSim is a relatively complex algorithm and does not provide insight about the updated state. 
We develop a simpler and faster alternative to GraphSim based on novel formulas enabling quicker gate updates.

Recently, the work of \cite{bruh} improves the runtime of applying controlled-Phase $Z$ gates of GraphSim by reducing the number of costly quadratic local complementation operations and optimizing controlled-Phase $Z$ gate updates to take linear time rather than quadratic time in certain cases. 
Our independently-discovered algorithm has similar runtimes and may run faster in practice. This is because it is simpler, directly computing the updated state upon controlled-Phase $Z$ gate application using our novel formulas instead of using casework and repeated local complementation.

With our novel formulas, we explore ways to improve the theoretical runtime of applying controlled-Phase $Z$ gates in graph state simulation.
Improving the runtime would enable graph state simulation to become the fastest method for simulating any type of stabilizer circuit \cite{kerzner}.
Currently, the runtime is believed to be $O(d^2)$ where $d$ is the maximum degree of the vertices in the graph \cite{kerzner, bruh}. 
We demonstrate that if we apply gates one at a time, in certain scenarios, the $\Omega(d^2)$ runtime is unavoidable.
Our algorithm is the optimal algorithm for applying gates in graph state simulation. 
Our work suggests that further improvements to graph state simulation must be based on the specifications of the circuit to be simulated.

In Section \ref{section4}, we develop a simpler algorithm for the merging of two stabilizer states related by a Pauli operator, which can be applied to compute measurements of Pauli observables. In \cite{multiparty}, the measurements of single-qubit Pauli operators on graph states are computed. The work of \cite{stabgraph2} and \cite{khesin ren} explores the multiple-qubit case, but while the former relies on a complicated procedure, the latter method utilizes a complicated formula in one case. Our work simplifies the formula in \cite{khesin ren}. 
We apply our graph merging formula to discover our formulas for applying controlled-Phase $Z$ gates to stabilizer states.

In addition, we study the linear dependence of stabilizer states by fully examining the special case of the linear dependence of three stabilizer states. In \cite{garcia,garcia2}, the rich geometric structure of stabilizer states is studied to design more efficient techniques for representing and manipulating quantum states. 
We extend the work of \cite{garcia2} by finding two new cases of linearly dependent triplets with symmetric pairwise inner products.

To improve efficiency in detecting the linear dependence of three stabilizer states and computing stabilizer states that are in the span of two given stabilizer states, among other applications, 
we develop an efficient inner product computation algorithm for stabilizer states represented with graph states. 
For stabilizer states represented with matrices, inner products can be computed in $O(n^3)$ \cite{garcia}. 
For stabilizer states represented with graph states, inner products can currently be computed in $O(Ed^2+nd^2)$ \cite{bruh}, where $E$ is the number of edges in the graph and $d$ is the maximum degree. 
By applying controlled-Phase $Z$ gates in a particular way, our algorithm improves runtime to $O(nd^2)$.
\section{The Stabilizer and Graph Formalisms}
\label{section2}
Here we define important notations used throughout the paper. We start by defining the operators that we use frequently in this paper.
Let a \textit{Pauli operator} $P$ on $n$ qubits be of the form $i^k \bigotimes\limits_{i=1}^{n}P_i$ where $k\in \{0,1,2,3\}$ and $P_i\in\{I,X,Y,Z\}$ is a Pauli matrix. The Pauli matrices are defined as $I\equiv 
\begin{pmatrix}
1&0\\
0&1
\end{pmatrix}$, $X\equiv
\begin{pmatrix}
0&1\\
1&0
\end{pmatrix}$, $Y\equiv \begin{pmatrix}
0&-i\\
i&0
\end{pmatrix}$, and $Z\equiv \begin{pmatrix}
1&0\\
0&-1
\end{pmatrix}$. Let the set of all Pauli operators be $\mathcal{P}$, the \textit{Pauli group}.

For some gate $U$, we let $CU_{a,b}$ denote the \textit{controlled-U gate} with control qubit $a$ and target qubit $b$. For example, we let $CX_{a,b}$ denote the \textit{controlled-X} gate with control qubit $a$ and target $b$, and we define $CY_{a,b}$ and $CZ_{a,b}$ similarly. We place subscripts on single-qubit operators to turn them into $n$-qubit operators where that operator is applied to the particular qubit referred to by the subscript, and $n$ is contextual. For example, $Z_1$ would be the Pauli $Z$ gate on qubit $1$. If $a=b$, then we assume $U$ is diagonal and let $CU_{a,a}\equiv U_a$.

Let a \textit{Clifford operator} $C$ on $n$ qubits be a unitary operator on $2^n$ dimensional state space such that for all Pauli operators $P$ on $n$ qubits, $CPC^{\dagger}\in \mathcal{P}$. Let the set of all Clifford operators be $\mathcal{C}$, the \textit{Clifford group}, which is generated by the \textit{Hadamard gate} $H\equiv \frac{1}{\sqrt{2}}\begin{pmatrix}
1&1\\
1&-1
\end{pmatrix}$, the \textit{phase gate} $S\equiv \begin{pmatrix}
1&0\\
0&i
\end{pmatrix}$, and any controlled-Pauli gate \cite{gottesman}.
We call the Clifford operators $C$ acting on a single qubit \textit{local Clifford operators}, and these operators are generated by $H$ and $S$ up to global phase.

Let the $n$-qubit state $\ket{\psi}$ be a \emph{stabilizer state} if there exists a set of $n$ commuting independent Pauli operators, $\{g_1,g_2,\dots,g_n\}$, such that for all $i\in \{1,2,\dots,n\}$, $g_i^2=I$ and $g_i\ket{\psi}=\ket{\psi}$. We call the operators $g_i$ \textit{stabilizers}. A stabilizer state $\ket{\psi}$ is equivalently defined as a state resulting from the action of a Clifford operator $C$ on a computational basis state \cite{chp}.

For a \textit{graph} $G$, we let $E(G)$ refer to the set of edges of $G$ and $V=\{1,2,\dots n\}\equiv \left[n\right]$, where vertex $i$ and qubit $i$ are synonymous \cite{graph}. We assume our graphs are undirected and do not have loop edges or multiple edges between two qubits. Let $\nbhd(i)$ be the set of neighbors of $i$ in $G$ not including $i$, where $G$ is contextual. 
Let the local complementation of a graph $G$ at qubit $i$, $L_i(G)$, be $G$ except that for each pair of qubits in $\nbhd(i)$, the corresponding edge is in $L_i(G)$ if and only if it is not in $G$.

The \textit{graph state} of a graph $G$, $\ket{G}$, is the stabilizer state with stabilizers $g_i\equiv X_i\prod\limits_{j\in \nbhd(i)}Z_j$ for $1\le i\le n$ \cite{van den nest}. 
An equivalent way of defining graph states is \begin{align}
    \ket{G}\equiv \left(\prod\limits_{(i,j)\in E(G)}\hspace{-2ex}CZ_{i,j}\right)\ket{+}^{\tens{}n},
\end{align} where $\ket{+}\equiv \frac{1}{\sqrt{2}}(\ket{0}+\ket{1})$ \cite{van den nest}. When $\ket{G}$ is expressed as a state vector, the global phase is fixed by assuming the amplitude of $\ket{0}^{\tens{}n}$ is positive and real.

We define our terminology for stabilizer states represented as applications of local Clifford operators to graph states, which is enabled by a theorem proved in \cite{van den nest}.
\begin{definition}
An \emph{extended graph state} is a graph state with local Clifford operators applied to it, written as $C\ket{G}$ where $C$ is a tensor product of local Clifford operators.
\end{definition}
Let the \textit{support} of a quantum state $\ket{\psi}$ be the number of non-zero amplitudes it has when written as a state vector, and let the \textit{support set} be the set of vectors corresponding to the computational basis states with non-zero amplitudes in $\ket{\psi}$.

These definitions enable us to examine stabilizer states from the perspective of graph states.

\section{Canonical forms for stabilizer states}
\label{section3}
\subsection{Canonical Generator Matrix}
The binary representation of the stabilizer formalism \cite{mikeike} associates a binary vector with each Pauli operator generator. 
The generators of an $n$-qubit stabilizer state are stored in an $n\times 2n$ generator matrix. 
The rows of the generator matrix are linearly independent, and a shifted inner product can be defined so that it is $0$ for all pairs of distinct rows of the generator matrix.
Swapping rows corresponds to swapping generators, adding a row to another corresponds to multiplying generators, and switching columns corresponds to swapping qubits. 
These operations can transform a generator matrix into a \textit{canonical form},
\begin{align}
\left(
    \begin{array}{cc|cc}
    I & A & B & 0\\
    0 & 0 & A^T & I
    \end{array}\right),
\end{align}where $B$ is symmetric. 
However, this canonical form is not unique because of the freedom in choosing how to swap qubits. 
Furthermore, this canonical form can be converted into the \textit{reduced form} for extended graph states \cite{stabgraph1}.
\begin{definition}
Let an extended graph state $C\ket{G}$ be in \emph{reduced form} if there exist $n$-tuples $c\equiv (c_1,\dots c_n)$ and $z\equiv (z_1,\dots z_n)$ with $c_i\in \{ I,S,H\}$ and $z_i\in \{I,Z\}$ such that $C=\bigotimes\limits_{i=1}^nc_iz_i$, and for all $(i,j)\in E(G)$, either $c_i\neq H$ or $c_j\neq H$.
\end{definition}
The reduced form provides an elegant graphical representation of stabilizer states, but multiple extended graph states in reduced form can refer to the same quantum state. 
\subsection{A unique canonical form}
Our canonical form is an extension of the reduced form.
\begin{definition}
Let an extended graph state in reduced form be in \emph{canonical form} if for all $(i,j)\in E(G)$ such that $c_i=H$, we have $j>i$.
\end{definition}
\begin{figure}
    \centering
    \begin{tikzpicture}
	    \begin{pgfonlayer}{nodelayer}
    		\node [label={[label distance=-0.1cm]30: $SZ$ },style=white dot] (0) at (0, 0) {$3$};
    		\node [style=white dot] (1) at (-2, 1) {$4$};
    		\node [label={[label distance=-0.1cm]30: $Z$ },style=white dot] (2) at (2, 2.75) {$5$};
    		\node [label={[label distance=-0.1cm]30: $S$ },style=white dot] (3) at (-0.1, 3.5) {$6$};
    		\node [label={[label distance=-0.1cm]15: $H$ },style=white dot] (4) at (-1, -1.25) {$1$};
    		\node [label={[label distance=-0.1cm]30: $H$ },style=white dot] (5) at (2.75, -0.25) {$2$};
    		\node [style=none] (6) at (-1.25, 3.75) {};
    		\node [label={[label distance=-0.1cm]30: $HZ$ },style=white dot] (7) at (-2.0, 3.9) {$7$};
    	\end{pgfonlayer}
    	\begin{pgfonlayer}{edgelayer}
    		\draw (1) to (0);
    		\draw (0) to (3);
    		\draw (5) to (0);
    		\draw (2) to (5);
    		\draw (4) to (0);
    		\draw (4) to (3);
	    \end{pgfonlayer}
    \end{tikzpicture}
    \caption{An illustration of the stabilizer state in canonical form, $\ket{\psi}=H_1H_2S_3Z_3Z_5S_6H_7Z_7\ket{G}$, where $E(G)=\{(1,3),(1,6),(2,3),(2,5),(3,4),(3,6)\}$.}
    \label{canonicalform}
\end{figure}
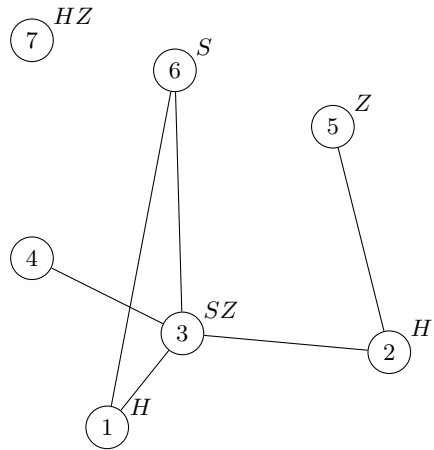
The following result relates the number of $H$'s to the support and helps us prove the canonical form is unique.
\begin{lemma}
\label{amp}
Let $\ket{\psi}=\bigotimes\limits_{i=1}^nc_iz_i\ket{G}$ be in reduced form. Let $k$ be the number of $c_i$ that are equal to $H$. Then the support of $\ket{\psi}$ is $2^{n-k}$.
\end{lemma}
\begin{proof}
Let $A\equiv \{i\in\left[n\right]|c_i=H\}$, where $k=|A|$. 
We use the identity $H_iCZ_{i,j}=CX_{j,i}H_i$. 
We also define single-qubit Pauli operators $p_i$ as $p_i\equiv c_iz_ic_i^{\dagger}$. Then
\begin{multline}
\label{equation3}
\bigotimes\limits_{i=1}^nc_iz_i\ket{G}=\bigotimes\limits_{i=1}^np_i\prod\limits_{i\in \left[n\right]\setminus A}(c_i)_i\prod\limits_{i\in A}H_i\ket{G}\\
=\bigotimes\limits_{i=1}^np_i\prod\limits_{i\in \left[n\right]\setminus A}(c_i)_i\prod\limits_{i\in A}\left(H_i\prod\limits_{j\in \nbhd(i)}CZ_{i,j}\right)\\
\cdot \prod\limits_{(i,j)\in E(G),i\not\in A,j\not\in A}CZ_{i,j}\ket{+}^{\tens{}n}\\
=\bigotimes\limits_{i=1}^np_i\prod\limits_{i\in \left[n\right]\setminus A}(c_i)_i\prod\limits_{i\in A}\left(\prod\limits_{j\in \nbhd(i)}CX_{j,i}\right)\\
\cdot\prod\limits_{(i,j)\in E(G),i\not\in A,j\not\in A}CZ_{i,j} \prod\limits_{i\in A}H_i\ket{+}^{\tens{}n}.
\end{multline}

$\prod\limits_{i\in A}H_i\ket{+}^{\tens{}n}$ has support $2^{n-k}$ because it consists of a tensor product of $k$ $\ket{0}$'s and $n-k$ $\ket{+}$'s. 
The rest of Equation \ref{equation3} is a product of phase operators, Pauli operators, and controlled-Pauli operators, which does not change the support of $\ket{\psi}$.
\end{proof}
The main advantage of our canonical form is that it uniquely represents a stabilizer state.
\begin{theorem}
If $\ket{\psi}=\ket{\psi'}$ up to global phase, and $\ket{\psi}\equiv \bigotimes\limits_{i=1}^nc_iz_i\ket{G}$ and $\ket{\psi'}\equiv \bigotimes\limits_{i=1}^nc_i'z_i'\ket{G'}$ are both in canonical form, then $G=G'$, $c=c'$, and $z=z'$.
\end{theorem}
\begin{proof}
Let $A\equiv \{i\in\left[n\right]|c_i=H\}$ and $A'\equiv \{i\in \left[n\right]|c_i'=H\}$. The supports of $\ket{\psi}$ and $\ket{\psi'}$ are equal, so by Lemma \ref{amp}, $|A|=|A'|\equiv k$. Now let $A=\{a_1,a_2,\dots a_k\}$ and $A'=\{a_1',a_2',\dots,a_k'\}$ where $a_1<a_2<\dots<a_k$, $a_1'<a_2'<\dots<a_k'$. First, important definitions.
\begin{definition}
For a binary string $s$ and a subset $B\subset \left[n\right]$ of the same size, let $\ket{s}\bra{s}_{B}$ be a $n$-qubit projector onto the subspace of $n$-qubit state space spanned by the basis of computational basis states that agree with $s$ on the qubits in $B$. 
We can think of $\ket{s}\bra{s}_{B}$ as stretching the bits in $s$ out in a $n$ dimensional vector to occupy the slots corresponding to qubits in $B$, and
we let $s_i$ denote the bit in $s$ in the slot corresponding to qubit $i$. 
\end{definition}
\begin{definition}
For a single-qubit state $\ket{\varphi}$ and a subset $B\subset \left[n\right]$, let $\ket{\varphi}_B$ be a tensor product of $|B|$ $\ket{\varphi}$'s, placing them in the slots corresponding to the qubits in $B$.
\end{definition}
Now, suppose for the sake of contradiction $a_1<a_1'$. We will apply projectors of the form $\ket{s}\bra{s}_{\left[n\right]\setminus A}$ and $\ket{s}\bra{s}_{\left[n\right]\setminus A'}$ to $\ket{\psi}$ and $\ket{\psi'}$ to derive a contradiction. 
Letting $Q\equiv \ket{s}\bra{s}_{\left[n\right]\setminus A}$, we write
\begin{align}
    Q\bigotimes\limits_{i=1}^nc_iz_i\ket{G}=\bigotimes\limits_{i=1}^nc_iz_i\prod\limits_{(i,j)\in E(G)}CZ_{i,j}Q\ket{+}^{\tens{} n}\label{eqn4}\\
    =\frac{1}{\sqrt{2^{n-k}}}\bigotimes\limits_{i=1}^nc_iz_i\prod\limits_{(i,j)\in E(G)}CZ_{i,j}\ket{+}_A\tens{}\ket{s} \label{eqn5}\\
    =\frac{1}{\sqrt{2^{n-|A|}}}\bigotimes\limits_{i=1}^nc_iz_i\prod\limits_{(i,j)\in E(G)}Z_j^{s_i}\ket{+}_A\tens{}\ket{s}\label{eqn6}\\
    =\frac{1}{\sqrt{2^{n-|A|}}}\bigotimes\limits_{i=1}^np_i\prod\limits_{i\in \left[n\right]\setminus A}(c_i)_i \ket{0}_{A}\tens{}\ket{s}\label{eqn7},
\end{align}where Line \ref{eqn4} follows from the fact that $Q$ commutes with $S$, $Z$, and $CZ$ operators, Line \ref{eqn5} follows from assuming without loss of generality that $i\not\in A$ due to there being no edges in $G$ 
between qubits in $A$, and Line \ref{eqn7} follows by conjugating the $Z$ operators to form Pauli operators $p_i$. Observe that
\begin{proposition}
\label{onestate}
For all binary strings $s$ of length $n-k$, $\ket{s}\bra{s}_{\left[n\right]\setminus A}\ket{\psi}$ and $\ket{s}\bra{s}_{\left[n\right]\setminus A'}\ket{\psi'}$ are computational basis states.
\end{proposition}
\begin{proof}
By Equation \ref{eqn7}, $\ket{s}\bra{s}_{\left[n\right]\setminus A}\ket{\psi}$ is a single computational basis state because it consists of Pauli 
and phase operators applied to a computational basis state, and similarly $\ket{s}\bra{s}_{\left[n\right]\setminus A'}\ket{\psi'}$ is as well.
\end{proof}
In particular, if we consider $Q_1\equiv \ket{u}\bra{u}_{\left[n\right]\setminus A'}$, where $u_{a_1}=1$ and the rest of the $u_i$ equal $0$, and $Q_2\equiv \ket{0}\bra{0}_{\left[n\right]\setminus A'}$, we have
\begin{claim}
\label{claim111}
$Q_1\ket{\psi'}$ and $Q_2\ket{\psi'}$ are non-zero computational basis states that differ only in qubit $a_1$.
\end{claim}
\begin{proof}
By Line \ref{eqn6}, 
\begin{align*}
    Q_1\ket{\psi'}&=\frac{1}{\sqrt{2^{n-k}}}\bigotimes\limits_{i=1}^nc_i'z_i'\prod\limits_{(i,j)\in E(G')}Z_j^{u_i}\ket{+}_{A'}\tens{}\ket{u}\nonumber\\
    &=\frac{1}{\sqrt{2^{n-k}}}\bigotimes\limits_{i=1}^nc_i'z_i'\prod\limits_{j\in N_{G'}(a_1)}Z_j\ket{+}_{A'}\tens{}\ket{u}\nonumber\\
    &=\frac{1}{\sqrt{2^{n-k}}}\prod\limits_{j\in N_{G'}(a_1)}Z_j\bigotimes\limits_{i=1}^nc_i'z_i'\ket{+}_{A'}\tens{}\ket{u},
\end{align*}
and
\begin{align*}
    Q_2\ket{\psi'}&=\frac{1}{\sqrt{2^{n-k}}}\bigotimes\limits_{i=1}^nc_i'z_i'\prod\limits_{(i,j)\in E(G')}Z_j^{u_i}\ket{+}_{A'}\tens{}\ket{0}_{\left[n\right]\setminus A'}\nonumber\\&=\frac{1}{\sqrt{2^{n-k}}}\bigotimes\limits_{i=1}^nc_i'z_i'\ket{+}_{A'}\tens{}\ket{0}_{\left[n\right]\setminus A'}.
\end{align*}
\end{proof}
Let $s''$ be the unique binary string such that \begin{equation*}
    \ket{s''}\bra{s''}_{\left[n\right]\setminus A}Q_1\ket{\psi'}=Q_1\ket{\psi'}.
\end{equation*} 
By Claim \ref{claim111}, $\ket{s''}\bra{s''}_{\left[n\right]\setminus A}Q_2\ket{\psi'}=Q_2\ket{\psi'}$. 
Then, the support of $\ket{s''}\bra{s''}_{\left[n\right]\setminus A}\ket{\psi'}$ is at least $2$. 
However, the support of $\ket{s''}\bra{s''}_{\left[n\right]\setminus A}\ket{\psi}$ is $1$ by Proposition \ref{onestate}. Since $\ket{\psi}=\ket{\psi'}$, this produces the desired contradiction.

Then, we must have $a_1=a_1'$. We cancel the $H$'s from both sides and reduce $k$ by $1$ until $k$ is $0$. 
For all $j$, we have $c_jz_j\in\{I,S,Z,SZ\}.$ If $c_{j}\neq c_{j}'$ or $z_j\neq z_j'$ for some $j$, then the amplitudes of $\ket{0}^{\tens{} j-1}\tens{}\ket{1}\tens{} \ket{0}^{\tens{}n-j}$ in $\ket{\psi}$ and $\ket{\psi'}$ would differ by some power of $i$ that is not $1$. Therefore, $c=c'$ and $z=z'$, and we have $\ket{G}=\ket{G'}$. If $(i,j)\in E(G)$ but $(i,j)\not\in E(G')$ or vice versa, the amplitudes of $\ket{0}_{\{i,j\}}\tens{}\ket{0}^{\tens{}n-2}$ differ. Thus, $G=G'$.
\end{proof}
We show that any stabilizer state can be expressed in our canonical form by a counting argument. 
 In \cite{chp}, it is proven that the number of $n$-qubit stabilizer states is \begin{equation}
    2^n\prod\limits_{k=1}^n(2^{k}+1).
\end{equation}Because our canonical form is unique, it suffices to show the following result.
\begin{lemma}
There are $2^n\prod\limits_{k=1}^n(2^{k}+1)$ $n$-qubit extended graph states in canonical form.
\end{lemma}
\begin{proof}
We wish to count all possible $c,z,$ and $G$ such that $c_i\in \{I,S,H\}$ and $z_i\in \{I,Z\}$ for all $i\in \left[n\right]$ and whenever $c_i=H$, all the edges in $G$ incident to $i$ connect to higher numbered qubits.
For each qubit $k$, we choose $c_k,z_k$, and all the edges of the form $(i,k)$ where $i<k$. 
If none of the $(i,k)\in E(G)$, then there are no restrictions on $c_k$ and $z_k$, yielding $6$ possibilities. 
Otherwise, the only restriction is $c_k\neq H$, and there are $2^{k-1}-1$ possible choices for the edges, yielding $4(2^{k-1}-1)$ possibilities. 
In total, there are $2^{k+1}+2$ ways to choose $c_k,z_k$, and all the edges of the form $(k,i)$ where $i<k$, and doing so for each $1\le k\le n$ yields all possible extended graph state in canonical form.
Thus, there are $\prod\limits_{k=1}^n(2^{k+1}+2)=2^n\prod\limits_{k=1}^n(2^{k}+1)$ $n$-qubit extended graph states in canonical form.
\end{proof}
\subsection{Simplifying extended graph states}
\label{simplifyalgo}
We demonstrate how to simplify extended graph states to canonical form. 
Like in \cite{stabgraph1}, we repeatedly apply two transformation rules. 
The first, used in \cite{multiparty,stabgraph1}, relates $\ket{L_i(G)}$ to $\ket{G}$.
\begin{theorem}[Van den Nest et al., Hein et al.]
\label{lc}
For any graph state $\ket{G}$ and qubit $x$, \begin{align}
\label{eqn1212}
    \ket{G}=H_xS_x^{\dagger}H_x\prod_{i\in \nbhd(x)}S_i\ket{L_x(G)}.
\end{align}
\end{theorem}
The second, discovered in \cite{stabgraph1}, allows us to simplify extended graph states to reduced form by eliminating pairs of $H$'s applied to connected qubits and also allows us to simplify to canonical form by sliding $H$'s down to smaller numbered qubits.
We present our own proof in Appendix \ref{appendixa} because it uses a different methodology.
\begin{theorem}[Elliot et al.]
\label{hsliding}
Let $(x,y)\in E(G)$. Let $A=N(x)\cup \{x\}$ and $B=N(y)\cup \{y\}$. Then, \begin{align}
\label{eqn1313}
    H_xH_y\ket{G}=Z_xZ_y\prod_{p\in A,q\in B}CZ_{p,q}\ket{G}.
\end{align}
\end{theorem}
The runtime of the algorithm, though cubic in the worst case, can be much quicker.
\begin{theorem}
\label{reachability}
There exists an algorithm to simplify an arbitrary extended graph state, $\ket{\psi}\equiv \bigotimes\limits_{i=1}^nC_i\ket{G}$, into canonical form, that runs in $O(nd^2)$, where $d$ is the maximum degree in $G$ encountered during the calculation.
\end{theorem}
\begin{proof}
Multiplying both sides of Equation \ref{eqn1212} by $S_xH_x$, we obtain 
\begin{align}
\label{1}
S_xH_x\ket{G}=H_x\prod_{p,q\in \nbhd(x)}CS_{p,q}\ket{G}.
\end{align}
Multiplying both sides of Equation \ref{1} by $H_xS_xH_xS_x^{\dagger}$, we obtain
\begin{align}
\label{3}
    H_xS_x\ket{G}&=H_xS_xH_xS_x^{\dagger}H_x\prod_{p,q\in \nbhd(x)}CS_{p,q}\ket{G}\nonumber\\
    &=\frac{1+i}{\sqrt{2}}S_1^{3}\prod_{p\in\nbhd(x)}Z_p\prod_{p,q\in \nbhd(x)}CS_{p,q}\ket{G}.
\end{align}
Because $C_i\in \langle H,S\rangle$, each $C_i$ is equivalent to a product of $H$'s and $S$'s up to global phase. Because $HH=I$ and $SS=Z$ are both Pauli operators, $C_i$ is equivalent to a global phase and a Pauli operator applied to an alternating product of $H$'s and $S$'s, which we define as $D_i$. 
Thus we can write $\ket{\psi}=\alpha P \bigotimes\limits_{i=1}^n D_i \ket{G}$ for some constant $\alpha\in \mathbb{C}\setminus\{0\}$ and some Pauli operator $P$. 
In what follows, we do not mention Pauli operators or global phases because we can automatically keep track of them by conjugating them through and updating $\alpha$ and $P$ accordingly. 
We define the useful monovariant and describe an algorithm to decrease it.
\begin{definition}
Let $M$ be the sum of the total number of $H$'s among all $D_p$ for $1\le p\le n$ and the number of $p$ such that $D_p$ ends in $SH$.
\end{definition}
\begin{lemma}
\label{algo}
 If $D_i$ has length at least $2$, we can update $D_i$ and all $D_j$ for $j\in \nbhd(i)$ so that $M$ decreases.
\end{lemma}
\begin{proof}
If $D_i$ ends in $HS$, we apply Equation \ref{3} on qubit $i$, which effectively removes $HS$ from $D_i$ and appends $S$ onto the ends of $D_i$ and all $D_j$. Otherwise, $D_i$ ends in $SH$. 
If some $D_j$ ends in $H$, we apply Theorem \ref{hsliding} with $x=i,y=j$ to remove $2$ $H$'s. 
The last case is if all $D_j$ end in $I$ or $S$, in which case applying Equation \ref{1} on qubit $i$ will change $D_i$ to not end in $SH$.
\end{proof}
Because $M\ge 0$, we can apply the updates in Lemma \ref{algo} a finite number of times until all $D_i$ have length at most $1$, in which case $D_i\in \{I,S,H\}$ for all $i$. Rearranging Equation \ref{eqn1313} and assuming $x>y$, we have
\begin{align}
\label{realhsliding}
    H_x\ket{G}=H_yZ_xZ_y\prod\limits_{p\in A,q\in B}CZ_{p,q}\ket{G},
\end{align}
which we repeatedly apply on qubits $x$ with $D_x$ ending in $H$ whenever $x$ has a neighbor $y$ in $G$ with $y<x$. 
This must terminate since all vertices are at integers at least $1$. 
Now $D_i\in \{I,S,H,SH\}$ for all $i$. For all $i$ such that $D_i=SH$, we apply Equation \ref{1} and simplify, not having to worry about $D_j$ having length greater than $1$ because $D_j$ cannot end in $H$ by assumption. 
After this terminates, we conjugate $P$ through. Using the fact that $X_i\ket{G}=\prod\limits_{p\in \nbhd(i)}Z_i \ket{G}$, and $Y=-iZX$, we can turn all $Y$'s and $X$'s into $Z$'s. Now, we have transformed $\ket{\psi}$ into our canonical form.

Every time we apply Theorem \ref{hsliding}, Equation \ref{1}, or Equation \ref{3}, we perform $O(d^2)$ edge toggles where $d$ is the maximum degree of $G$. 
Initially, $M=O(n)$ because any local Clifford operator can be represented with a finite number of $H$'s. Then, shortening the lengths of all the $D_i$ to $1$ takes $O(nd^2)$ operations. 
Next, we only need to apply Equation \ref{realhsliding} at most $n-1$ times by applying it for $x=n,n-1,n-2,\dots ,2$ in that order. Because the $H$'s move to lower numbered qubits or are eliminated, the only way an $H$ could still exist on a qubit $p$ after the algorithm passes through $p$ the first time is if right before the algorithm passes through $p$, all $q\in N(p)$ satisfy $q>p$. 
In that case, $N(p)$ cannot change once $x\le p$, because $p$ is not connected to any lower numbered qubits.
Therefore, after $x$ reaches $2$, none of the $H$'s can be moved to lower numbered qubits. 
Thus, moving the $H$'s to the lowest possible numbered qubits takes $O(nd^2)$ operations. Removing all $D_i$ that equal $SH$ using Equation \ref{1} takes $O(nd^2)$ operations, and simplifying the Pauli operators takes $O(nd)$ operations, so the total runtime is $O(nd^2)$.
\end{proof}
\section{Graph state stabilizer simulation}
\label{section5}
\subsection{Algorithm}
Graph state simulators of stabilizer circuits are advantageous in that local Clifford gates such as $S$ and $H$ can be applied trivially in $O(1)$ time. 
The bottleneck of a graph state simulator is the application of controlled-Pauli gates, such as $CZ$ gates, which currently can be done in $O(d^2)$ time where $d$ is the maximum degree of the graph encountered during the calculation.

To apply a gate $CZ_{x,y}$ to an extended graph state $\ket{\psi}=\bigotimes\limits_{i=1}^nC_i\ket{G}$, we use the identity \begin{equation*}
    CZ_{x,y}=\frac{1}{2}\left((I+Z_x)+(I-Z_x)Z_y\right),
\end{equation*} conjugating the expression through the $C_i$ so it suffices to apply operators of the form $\frac{1}{2}((I+P_x)+(I-P_x)Q_y)$ to graph states $\ket{G}$ where $P_x$ and $Q_y$ are Hermitian Pauli operators. This motivates the following definition.
\begin{definition}
For Hermitian Pauli operators $P$ and $Q$, let $\ket{\psi_{PQ}}$ be the extended graph state obtained from simplifying the expression
\begin{equation}
    \frac{1}{2}\left( (I+P_1)+(I-P_1)Q_2 \right)\ket{G}.
\end{equation}
For example, $\ket{\psi_{ZZ}}$ is $CZ_{1,2}\ket{G}$, so updating $G$ takes $O(1)$ time.
\end{definition}
Our expressions for $\ket{\psi_{PQ}}$ and the update times based on the expressions are depicted in Table \ref{table1}.
\begin{table*}
\begin{tabular}{|c|c|c|c|}
\hline
\multirow{2}{*}{$(P,Q)$} & \multicolumn{2}{c|}{$\ket{\psi_{PQ}}$}       & \multirow{2}{*}{Update time} \\ \cline{2-3}
                       & $(1,2) \in E(G)$ & $(1,2) \not\in E(G)$ &                              \\ \hline
$(Z,Z)$                  & \multicolumn{2}{c|}{$\vphantom{\left(CZ_{1,2}\ket{G}\right)}CZ_{1,2}\ket{G}$}            & \multirow{3}{*}{$O(d)$}           \\\cline{1-3}
$(Z,X)$                  & \multicolumn{2}{c|}{$\vphantom{\left(\prod\limits_{x\in N_2}CZ_{1,x}\ket{G}\right)}\prod\limits_{x\in N_2}CZ_{1,x}\ket{G}$}            &                              \\ \cline{1-3}
$(Y,Z)$                  & \multicolumn{2}{c|}{$\vphantom{\left(S_2Z_2\prod\limits_{x\in M_1}CZ_{2,x}\ket{G}\right)}S_2Z_2\prod\limits_{x\in M_1}CZ_{2,x}\ket{G}$}            &                              \\ \hline
$(X,X)$                  & $\vphantom{\left(H_1H_2CZ_{1,2}\prod\limits_{x\in M_1,y\in M_2}CZ_{x,y}\ket{G}\right)}H_1H_2CZ_{1,2}\prod\limits_{x\in M_1,y\in M_2}CZ_{x,y}\ket{G}$             & $\prod\limits_{x\in N_1,y\in N_2}CZ_{x,y}\ket{G}$                 & \multirow{3}{*}{$O(d^2)$}           \\[0.5ex] \cline{1-3}
$(Y,X)$                  & $\vphantom{\left(\frac{1-i}{\sqrt{2}}\prod\limits_{x\in M_1}S_x\hspace{0.5em} H_1 \prod\limits_{x\in M_1\triangle M_2}CZ_{1,x}\ket{L_1(G)}\right)}\frac{1-i}{\sqrt{2}}\prod\limits_{x\in M_1}S_x\hspace{0.5em} H_1 \prod\limits_{x\in M_1\triangle M_2}CZ_{1,x}\ket{L_1(G)}$             & $\prod\limits_{x\in M_1\triangle N_2}Z_x\prod\limits_{x,y\in M_1\triangle N_2}CS_{x,y}\prod\limits_{x,y\in M_1}CS_{x,y}\ket{G}$                 &                              \\ \cline{1-3}

$(Y,Y)$                  & $-i\prod\limits_{x,y\in M_1}CS_{x,y}\prod\limits_{x,y\in M_2}CS_{x,y}\ket{G}$             & $\vphantom{\left(\frac{1-i}{\sqrt{2}}\prod\limits_{x\in M_1}S_x H_1 \prod\limits_{x\in M_2}CZ_{1,x}\ket{L_1(G)}\right)}\frac{1-i}{\sqrt{2}}\prod\limits_{x\in M_1}S_x\hspace{0.5em} H_1 \prod\limits_{x\in M_2}CZ_{1,x}\ket{L_1(G)}$                 &                              \\ \hline
\end{tabular}
\caption{A table of formulas for $\ket{\psi_{PQ}}$, where $d=\max(\deg(1),\deg(2))$ and $\triangle$ is the symmetric difference of two sets. With this data, we can compute $\ket{\psi_{PQ}}$ for all possible unordered pairs $(P,Q)$ since $\ket{\psi_{PQ}}$ and $\ket{\psi_{QP}}$ are equal with the roles of qubits $1$ and $2$ flipped, and changing the sign of $P$ changes $\ket{\psi_{PQ}}$ by $Q$. Also, $N_1\equiv \nbhd(1)$, $N_2\equiv \nbhd(2)$, $M_1\equiv N_1\cup \{1\}$, and $M_2\equiv N_2\cup \{2\}$. Note that for $\{P,Q\}\in \{\{Z,Z\},\{Z,X\},\{Y,Z\}\}$, $\ket{\psi_{PQ}}$ consists of $O(d)$ $CZ$ operators applied to $\ket{G}$, whereas for $(P,Q)\in \{\{X,X\},\{Y,X\},\{Y,Y\}\}$, $\ket{\psi_{PQ}}$ consists of $O(d^2)$ $CZ$ operators and $O(d)$ local Clifford operators applied to $\ket{G}$, hence the $O(d^2)$ update time.}
\label{table1}
\end{table*}
When $(P,Q)\in \{(Z,X),(X,X)\}$, formulas for $\ket{\psi_{PQ}}$ were computed in \cite{stabgraph1}. The rest are our own discoveries. We computed these formulas by applying Theorem \ref{khesin} and Theorem \ref{graphmerging}. Since all the proofs are similar, they can be found in Appendix \ref{appendixc}.  
\subsection{Discussion}
To apply a $CZ$ gate to qubits $x$ and $y$ of the extended graph state $\ket{\psi}=\prod\limits_{i=1}^nC_i\ket{G}$, GraphSim \cite{graphsim,kerzner}, the currently widely adopted algorithm, performs local complementations on $x$, $y$, or neighboring qubits of $x$ and $y$, changing $C_x$ and $C_y$ until they are both diagonal. 
Local complementations run in $\Omega(d^2)$, where $d$ is the degree of the vertex at which it was applied. 
When applying a $CZ$ gate using our algorithm, if $P=\pm Z$ or $Q=\pm Z$, then it takes $O(d)$ time and runs much faster than GraphSim.
For example, when $C_x=HSH$ and $C_y=I$, GraphSim would perform a local complementation at qubit $x$, whereas our algorithm would update $\ket{\psi}$ more efficiently, based on the expression in Table \ref{table1} for $\ket{\psi_{YZ}}$.

Because $P$ and $Q$ are each equally likely to be any of $\{\pm X,\pm Y,\pm Z\}$ during a simulation of a quantum circuit, our algorithm outperforms GraphSim approximately $\frac{5}{9}$ of the time, leading to a significant efficiency advantage when $d$ becomes large.

In order to perform $CZ$ updates in under quadratic time, we must find efficient update rules for $\ket{\psi_{PQ}}$ for all multi-sets $\{P,Q\}\in \{\{X,X\},\{Y,Y\},\{X,Y\}\}$. 
We believe such update rules cannot directly be derived by applying the graph state transformation rules, Theorems \ref{lc} and \ref{hsliding}, to the expressions for $\ket{\psi_{XX}}$, $\ket{\psi_{YX}}$, and $\ket{\psi_{YY}}$ in Table \ref{table1} because there will always be edge toggles between two sets of vertices of size $O(d)$.
In fact, we show that finding such update rules is impossible if they update the graph by toggling its edges.
\begin{theorem}
\label{togglinghard}
There exists a family of extended graph states such that applying a $CZ$ gate requires $\Omega(n^2)$ edges of $G$ to be toggled.
\end{theorem}
\begin{proof}
Let $A\subset \left[n\right]$ and $B\equiv \left[n\right]\setminus A$. Let $1\in A$, $2\in B$, $|A|=\Omega(n)$, and $|B|=\Omega(n)$. The following graphs are used in this proof.
\begin{definition}
Let $W_{i,j}$, where $i\in A$ and $j\in B$, be the graph consisting solely of edges incident to either vertex $i$ or vertex $j$, such that vertex $i$ is connected to vertex $j$, vertex $i$ is connected to all vertices in $A\setminus \{i\}$, and vertex $j$ is connected to all vertices in $B\setminus \{j\}$. Let $K$ be the complete bipartite graph with edges between each vertex in $A$ and each vertex in $B$. Let $K_a$ (resp.~$K_b$) be $K$ together with all the edges between vertices in $A$ (resp.~$B$). Let $K_{a,i}$ (resp.~$K_{b,i}$) be the graph where all the vertices in $A$ (resp.~$B$) are connected to each other, and vertex $i$ is connected to every other vertex. Let $G$ be the graph that has all possible edges, except those between vertex $2$ and the vertices in $A$.
\end{definition}
Suppose we want to apply $CZ_{x,y}$ to an extended graph state $\ket{\psi}\equiv \bigotimes\limits_{i=1}^nC_i\ket{G}$ where $C_x=C_y=HSX$. 
Then, $CZ_{x,y}\ket{\psi}=\bigotimes\limits_{i=1}^nC_i\ket{\psi_{YY}}$. When we update $\ket{G}$, regardless of what algorithm we use, we end up with $\bigotimes\limits_{i=1}^nC_i'\ket{G'}$ for some $C_i'$ and $G'$ where $\ket{G'}$ is local Clifford equivalent to $\ket{\psi_{YY}}$. 
Applying Lemma \ref{yy},
\begin{align*}
    \ket{\psi_{YY}}=\frac{1-i}{\sqrt{2}}\prod\limits_{x\in \left[n\right] \setminus \{2\}} S_xH_1\ket{W_{1,2}},
\end{align*}so $G'$ is local Clifford equivalent to $\ket{W_{1,2}}$.
\begin{lemma}
\label{lotsofgraphs}
Let $R$ be the set of all graphs $G'$ that are local Clifford equivalent to $\ket{W_{1,2}}$. Then
\begin{multline}
R=\{K, K_a, K_b\}\cup \{K_{a,i}|i\in B\}\cup \{K_{b,i}|i\in A\}\\ \cup \{W_{i,j}| i\in A,j\in B\}.
\end{multline} 
\end{lemma}
\begin{proof}
By a theorem proved in \cite{van den nest}, the local Clifford equivalence of the graph states $\ket{G'}$ and $\ket{W_{1,2}}$ is equivalent to the existence of a sequence of local complementation operations taking $G'$ to $W_{1,2}$. 
If we let $\mathcal{G}$ be the connected graph of graphs containing $W_{1,2}$ where edges are drawn between two graphs related by a local complementation, then $R=V(\mathcal{G})$.
$\mathcal{G}$ is depicted in Figure \ref{diagram}.
The rest of the proof details traversing $\mathcal{G}$.
For all $i\in A$ and $j\in B$,
\begin{itemize}
\item We consider all the edges in $\mathcal{G}$ emanating from $W_{i,j}$. For all $k\in \left[n\right]\setminus\{i,j\}$, \begin{align*}
    L_i(W_{i,j})&=K_{a,j}\nonumber \\
    L_j(W_{i,j})&=K_{b,i}\nonumber\\
    L_k(W_{i,j})&=W_{i,j}
\end{align*}
\item We consider all the edges in $\mathcal{G}$ emanating from $K_{a,j}$. The case for $K_{b,i}$ is similar. For all $k\in B\setminus \{ j\}$,
\begin{align*}
    L_i(K_{a,j})&=W_{i,j}\nonumber\\
    L_j(K_{a,j})&=K_b\nonumber\\
    L_k(K_{a,j})&=K_{a,j}
\end{align*}
\item We consider all the edges in $\mathcal{G}$ emanating from $K$, $K_a$, or $K_b$.
\begin{align*}
    L_i(K_a)&=K_{b,i}\nonumber\\
    L_j(K_a)&=K\nonumber\\
    L_j(K_b)&=K_{a,j}\nonumber\\
    L_i(K_b)&=K\nonumber\\
    L_i(K)&=K_b\nonumber\\
    L_j(K)&=K_a
\end{align*}
\item The graph $W_{1,2}$ is connected to $W_{i,j}$ in $\mathcal{G}$.
\begin{equation*}
    L_j(L_2(L_i(L_1(W_{1,2}))))=W_{i,j}
\end{equation*}
\end{itemize}
\end{proof}
\begin{figure}
\label{diagram}
\centering
\begin{tikzpicture}[node distance={20mm}, thick, main/.style = {draw, circle,minimum size=3em}] 
\node[main] (1) {$W_{i,j}$}; 
\node[main] (2) [above right of=1] {$K_{a,j}$}; 
\node[main] (3) [below right of=1] {$K_{b,i}$}; 
\node[main] (4) [right of=2] {$K_b$}; 
\node[main] (5) [right of=3] {$K_a$}; 
\node[main] (6) [below right of=4] {$K$}; 
\draw (1) -- (2); 
\draw (1) -- (3); 
\draw (2) -- (4); 
\draw (3) -- (5); 
\draw (4) -- (6); 
\draw (5) -- (6); 
\draw (6) -- node[midway, above right, sloped, pos=0.6] {i} (4); 
\draw (6) -- node[midway, below right, sloped, pos=0.6] {j} (5); 
\draw (4) -- node[midway, above, sloped, pos=0.5] {j} (2); 
\draw (2) -- node[midway, above left, sloped, pos=0.4] {i} (1);
\draw (1) -- node[midway, below left, sloped, pos=0.5] {j} (3); 
\draw (3) -- node[midway, below, sloped, pos=0.6] {i} (5); 
\end{tikzpicture} 
\caption{A depiction of $\mathcal{G}$ in the proof of Lemma \ref{lotsofgraphs}, with undirected edges labeled with the vertex that local complementation is applied to and loop edges omitted. To generate $\mathcal{G}$ in its entirely, let $i$ and $j$ range over all vertices in $A$ and in $B$ respectively.}
\end{figure}
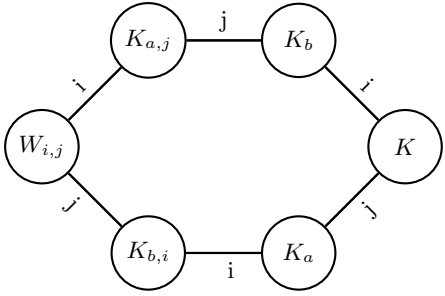
We show that $|E(G)\triangle E(G')|=\Omega(n^2)$ for any $G'\in R$.
Suppose without loss of generality that $\frac{1}{2}n\le |A|\le cn$ where $c$ is some fixed constant less than $1$. Then $K_a$ has $\binom{n}{2}-\binom{|B|}{2}$ edges, which is the most number of edges out of all graphs in $R$. 
\begin{align*}
    |E(G)\triangle E(G')|&\ge |E(G)|-|E(G')|\nonumber\\
    &\ge \left(\binom{n}{2}-|A|\right)-\left(\binom{n}{2}-\binom{|B|}{2}\right)\nonumber\\
    &\ge \binom{|B|}{2}-|A|.
\end{align*} $|E(G)\triangle E(G')|=\Omega(n^2)$ since $|B|\ge (1-c)n=\Omega(n)$.
\end{proof}
\section{Additive properties of stabilizer states}
\label{section4}
\subsection{Graph merging}
We first consider the case of two states related by a Pauli operator. The case when the Pauli operator acts on a single qubit was explored in \cite{multiparty}, and the case when the Pauli operator acts on multiple qubits was explored in \cite{khesin ren,stabgraph2}. 
We state the main theorem in \cite{khesin ren} here. In \cite{stabgraph2} a related theorem is proven but without the case where $k$ is odd.
\begin{theorem}[Khesin, Ren]
\label{khesin}
Let $A=\nbhd(1)\cup \{1\}$, and let $B$ be a set including $1$. Let $k$ be an integer. Then \begin{multline}
\frac{1}{\sqrt{2}}\left(I+i^k\prod_{j\in B}Z_j\right)\ket{G}\\=H_1Z_1\prod_{x\in A,y\in A}CS_{x,y}^k\prod_{x\in A,y\in B}CZ_{x,y}\ket{G}.
\end{multline}
\end{theorem}
We provide an alternative formula for when $k$ is odd that is more concise than previous formulas.
\begin{theorem}
\label{graphmerging}
Let $k=2m+1$. 
Let $A$ be an arbitrary set. Then \begin{multline}
    (I+i^{2m+1}\prod_{p\in A} Z_p)\ket{G}=(1+i^{2m+1}) \\ \cdot \prod_{p\in A}Z_p^{m+1}\prod_{p,q\in A}CS_{p,q}\ket{G}.
\end{multline}
\begin{proof}
Let $\ket{z}$ be some computational basis state, and let $r$ be the number of $i$ in $A$ where the $i$th bit in $z$ is $1$.
Let $f(r)=1$ when $r\equiv 2 \pmod{4}$ or $r\equiv 3\pmod{4}$ and $0$ otherwise and $g(r)=1$ when $r$ is odd and $0$ otherwise. Then,
\begin{multline*}
\bra{z}(I+i^{2m+1}\prod_{p\in A} Z_p)\ket{G}\\=\langle z|G\rangle + i^{2m+1}(-1)^r \langle z|G\rangle,
\end{multline*}
and
\begin{multline*}
\bra{z}(1+i^{2m+1})\prod_{p\in A}Z_p^{m+1}\prod_{p,q\in A}CS_{p,q}\ket{G}\\=(1+i^{2m+1})(-1)^{f(r)}(-1)^{(m+1)r}i^r\ket{G}\\
=(1+i^{2m+1})i^{g(r)}(-1)^{(m+1)r}\ket{G}.
\end{multline*}
The two expressions are equal for $m\in\{0,1\}$ and all $r$.
\end{proof}
\end{theorem}
The merging formulas, Theorem \ref{khesin} and Theorem \ref{graphmerging}, can be used to compute measurements of Pauli operators on extended graph states, by conjugating Pauli projectors through the local Clifford operators. These formulas can also be used to prove the correctness of the expressions for $\ket{\psi_{PQ}}$ in Table \ref{table1}, which we do in Appendix \ref{appendixc}.

Considering ways to merge stabilizer states that are not related by a Pauli operator, an interesting formula arises when $x$ and $y$ are not connected in Theorem \ref{hsliding}. 
\begin{theorem}
\label{hsplitting}
Let $x$ and $y$ be two vertices of $G$ that are not connected. Let $A=\nbhd(x)\cup\{x\}$ and $B=\nbhd(y)\cup\{y\}$. Then
\begin{multline}H_xH_y\ket{G}=\\Z_xZ_y\ket{G}+\prod\limits_{p\in\nbhd(x)}Z_p\prod\limits_{q\in\nbhd(y)}Z_q\prod_{x\in A,y\in B}CZ_{x,y}\ket{G}.
\end{multline}
\end{theorem}
The proof is technical and included in Appendix \ref{appendixa}.
\subsection{Linearly Dependent Triplets}
We now turn our attention to characterizing linearly dependent triplets of stabilizer states.
The following theorem shows that there are three types.
\begin{theorem}
\label{generalmerge}
Let $\mathcal{S}\equiv \{\ket{\psi_1},\ket{\psi_2},\ket{\psi_3}\}$ be a set of linearly dependent stabilizer states that are not all parallel. Then, up to global phase, one of the three cases must be true
\begin{enumerate}
    \item For some stabilizer state $\ket{\phi}$ and some Pauli operator $P$, \begin{equation}
        \mathcal{S}=\{\ket{\phi},P\ket{\phi},\frac{I+P}{\sqrt{2}}\ket{\phi}\}.
    \end{equation}
    \item 
    \label{case2}For some Clifford operator $C$, $1\le x\le n$, and an extended graph state in reduced form $\ket{\psi}$ such that $x$ is the only value of $i$ such that $c_i\neq H$ and $z_i=I$ whenever $c_i=H$, 
    \begin{equation}
       \label{eqnn2} \mathcal{S}=\{C\ket{0^n},C\ket{\psi},C\left(S_x\ket{\psi}\right)\}.
    \end{equation}
    \item \label{case3}For some Clifford operator $C$, $1\le x<y\le n$, and an extended graph state in reduced form $\ket{\psi}$ such that $x$ and $y$ are the only two values of $i$ such that $c_i\neq H$ and $z_i=I$ whenever $c_i=H$, 
    \begin{equation}
       \label{eqnn3} \mathcal{S}=\{C\ket{0}^{\tens{}n},C\ket{\psi}, C\left(Z_xZ_yCZ_{x,y}\ket{\psi}\right)\}.
    \end{equation}
\end{enumerate}
\end{theorem}
\begin{proof}
Let $U$ be a Clifford operator such that $\ket{\psi_1}=U\ket{0}^{\tens{}n}$. Let $\ket{\psi}\equiv U^{\dagger}\ket{\psi_2}$ and $\ket{\varphi}\equiv U^{\dagger}\ket{\psi_3}$. Any stabilizer state can be represented up to global phase as \begin{equation*}
    \frac{1}{\sqrt{|V|}}\sum_{x\in V}i^{l(x)}(-1)^{q(x)}\ket{x},
\end{equation*} where $V$ is an affine subspace of $\mathbb{F}_2^n$, $\ell(x)$ is a linear binary function on $n$ bits, and $q(x)$ is a quadratic binary function on $n$ bits. Let $(V,\ell(x),q(x))$ be the corresponding triple for $\ket{\psi}$. Without loss of generality let the first non-zero amplitudes in $\ket{\psi}$ and $\ket{\varphi}$ be positive real numbers. The linear dependence of the state vectors in $\mathcal{S}$ is equivalent to the existence of $\alpha,\beta \in \mathbb{C}\setminus \{0\}$ such that \begin{equation*}
    \frac{1}{\sqrt{|V|}}\ket{0}^{\tens{}n}+\alpha\ket{\psi}=\beta\ket{\varphi}.
\end{equation*}
Note that $|V|$ is a power of $2$. If $|V|=1$, $\ket{\psi}$ is a non-zero computational basis state. Since the non-zero amplitudes in stabilizer states differ from each other by powers of $i$, $\alpha$ must be a power of $i$,  $\alpha\ket{\psi}$ and $\ket{0^n}$ are Pauli related, and $\ket{\varphi}= \frac{\ket{0^n}+\alpha\ket{\psi}}{\sqrt{2}}$.

From now on assume $|V|>1$. Then $0^n\in V$ or else $\ket{\varphi}$ would have $|V|+1\neq 2^m$ non-zero amplitudes and could not be a stabilizer state.
Also note the support set of $\ket{\varphi}$ is either $V$ or $V\setminus\{0^n\}$, and $|V|-1\neq 2^m\quad\forall m\ge2$. Therefore, the only case when the support set of $\ket{\varphi}$ is $V\setminus\{0^n\}$ is if $|V|=2$ and $\alpha=-1$, in which case $\ket{\varphi}$ is a computational basis state, related by a Pauli operator to $\ket{0}^{\tens{}n}$.

From now on the support set of $\ket{\varphi}$ is $V$. Then, $\beta=1+\alpha$ and by comparing non-zero amplitudes of the left and right hand sides, $\frac{\alpha}{1+\alpha}=i^k$ for some $k\in \{1,2,3\}$. 
\begin{claim}
\label{claim1}
If $|V|\ge 8$ and $k=2$, it is not possible for $\ket{\varphi}$ to be a stabilizer state.
\end{claim}
\begin{proof}
Suppose $\ket{\varphi}$ was a stabilizer state. We consider the stabilizer state $\ket{\phi}$ with support set $V$ and quadratic and linear functions equal to the difference of the quadratic and linear functions of $\ket{\varphi}$ and $\ket{\psi}$. 
The un-normalized amplitudes of $\ket{\phi}$ are equal to the ratios of the amplitudes of $\ket{\varphi}$ and $\ket{\psi}$, which are $i^k$ for all non-zero computational basis states and $1$ for $\ket{0^n}$. We use the following proposition to derive a contradiction.
\begin{proposition}
\label{pauliexpect}
Let $\ket{\phi}$ be a stabilizer state. Then for any Pauli operator $P$, $\bra{\phi}P\ket{\phi}\in \{0,1,i,-1,-i\}$.
\end{proposition}
\begin{proof}
Let $\ket{\phi}=C\ket{0^n}$ for some Clifford operator $C$. Then, for some Pauli operator $P'$, $\bra{\phi}P\ket{\phi}=\bra{0^n}C^{\dagger}PC\ket{0^n}=\bra{0^n}P'\ket{0^n}\in \{0,1,i,-1,-i\}$.
\end{proof}
That $0^n\in V$ implies $V$ is a subspace of $\mathbb{F}_2^n$. Let $e\equiv e_1e_2\dots e_n$ be a basis vector of $V$. Let $P\equiv \bigotimes\limits_{i=1}^nX_i^{e_i}$. Then $P\ket{\phi}=\frac{1}{\sqrt{|V|}}\left(\ket{e}-\sum_{x\in V\setminus\{e\}}\ket{x}\right)$, so $\bra{\phi}P\ket{\phi}=\frac{|V|-4}{|V|}\not\in \{0,1,i,-1,-i\}$, contradicting Proposition \ref{pauliexpect}.
\end{proof}
\begin{claim}
\label{claim2}
If $|V|\ge 4$ and $k\in \{1,3\}$, it is not possible for $\ket{\varphi}$ to be a stabilizer state.
\end{claim}
\begin{proof}
As in Claim \ref{claim1}, define $\ket{\phi}$ equal to the stabilizer state whose un-normalized amplitudes are the ratios of the amplitudes of $\ket{\varphi}$ and $\ket{\psi}$, in which case $\ket{\phi}\propto \ket{0^n}\pm i \sum\limits_{x\in V\setminus\{0^n\}}\ket{x}$. 
It is known that in a stabilizer state with its first non-zero amplitude positive and real, the number of pure imaginary amplitudes must be $0$ or half of the support. $\ket{\phi}$ does not satisfy this condition, the desired contradiction.
\end{proof}

By Claims \ref{claim1} and \ref{claim2}, the remaining cases either satisfy $|V|=2$ or $|V|=4$ and $k=2$. 
If $|V|=2$, then $\ket{\psi}$ and $\ket{\varphi}$ are of the form $\ket{\psi}=\frac{\ket{0^n}+i^h\ket{s}}{\sqrt{2}}$ and $\ket{\varphi}=\frac{\ket{0^n}+i^{k+h}\ket{s}}{\sqrt{2}}$ for some $h$ and computational basis state $\ket{s}$. Also, $\ket{0^n}=-\sqrt{2}\alpha \ket{\psi}+\sqrt{2}(1+\alpha)\ket{\varphi}$. 
If $k=2$, then $\alpha=-\frac{1}{2}$ and $-2\alpha\ket{\psi}$ and $2(1+\alpha)\ket{\varphi}$ are Pauli related stabilizer states such that their sum divided by $\sqrt{2}$ is $\ket{0^n}$. If $k=1$, then $\alpha = \frac{i-1}{2}$. If we express $\ket{\psi}$ in reduced form, then $n-1$ of the $c_i$ are equal to $H$ by Lemma \ref{amp}, and we can let $x$ be the unique index $i$ such that $c_i\neq H$. By Proposition \ref{inner}, since $\bra{0^n}\psi\rangle\neq0$, for each $c_i=H$, we have $z_i=I$. Note that $s_x=1$ by Lemma \ref{onestate}, so $S_x\ket{\psi}=\ket{\varphi}$, and we have \begin{equation*}
    \ket{0^n}=\frac{1-i}{\sqrt{2}}\ket{\psi}+\frac{1+i}{\sqrt{2}}S_x\ket{\psi},
\end{equation*}
which corresponds to Case \ref{case2}. If $k=3$, then similar arguments yield the same result with the roles of $\ket{\psi}$ and $\ket{\varphi}$ swapped.

If $|V|=4$ and $k=2$, then $\alpha=-\frac{1}{2}$ and $\beta=\frac{1}{2}$. If we express $\ket{\psi}$ in reduced form, then $n-2$ of the $c_i$ are equal to $H$ by Lemma \ref{amp}, and we can let $x$ and $y$ be the indices $i$ such that $c_i\neq H$. By Proposition \ref{inner}, since $\bra{0^n}\psi\rangle\neq0$, for each $c_i=H$, we have $z_i=I$. By Lemma \ref{onestate}, we can write the computational basis states in $\ket{\varphi}$ and $\ket{\psi}$ as $\ket{i}_x\tens{}\ket{j}_y\tens{} \ket{s_{ij}}$ for $i,j\in \{0,1\}$ and binary strings of length $n-2$ $s_{ij}$. We compute
\begin{multline*}
    \bra{i}_x\tens{}\bra{j}_y\tens{}\bra{s_{ij}} Z_xZ_yCZ_{x,y}\ket{\psi}\\=(-1)^{1-(1-i)(1-j)}\bra{i}_x\tens{}\bra{j}_y\tens{}\bra{s_{ij}}\hspace{0.2ex}\ket{\psi},
\end{multline*}
so we have \begin{equation*}
    \ket{0}^{\tens{}n}=\ket{\psi}+Z_xZ_yCZ_{x,y}\ket{\psi},
\end{equation*}which corresponds to Case \ref{case3}.
\end{proof}
\begin{example}
Small illustrative examples of each of the three cases in Theorem \ref{generalmerge} are shown. Each of the stabilizer states is in canonical form with vertex $1$ being the lowest node in the diagram and vertex $3$ being the highest.
\begin{align}
    \resizebox{2cm}{!}{
\begin{tikzpicture}[baseline={([yshift=-.5ex]current bounding box.center)}]
	\begin{pgfonlayer}{nodelayer}
		\node [style=white dot] (0) at (-0.25, 1) {$H$};
		\node [style=white dot] (1) at (-0.75, -0.5) {$H$};
		\node [style=white dot] (2) at (0.75, 0) {$H$};
		\node [style=none] (3) at (1, 1.5) {};
		\node [style=none] (4) at (2, 0.25) {};
		\node [style=none] (5) at (1, -1) {};
		\node [style=none] (6) at (-1.5, 1.5) {};
		\node [style=none] (7) at (-1.5, -1) {};
	\end{pgfonlayer}
	\begin{pgfonlayer}{edgelayer}
		\draw (6.center) to (7.center);
		\draw (3.center) to (4.center);
		\draw (4.center) to (5.center);
	\end{pgfonlayer}
\end{tikzpicture}}&=
\frac{1}{\sqrt{2}}
\resizebox{2cm}{!}{
\begin{tikzpicture}[baseline={([yshift=-.5ex]current bounding box.center)}]
	\begin{pgfonlayer}{nodelayer}
		\node [style=white dot] (0) at (-0.25, 1) {$H$};
		\node [style=white dot] (1) at (-0.75, -0.5) {$I$};
		\node [style=white dot] (2) at (0.75, 0) {$H$};
		\node [style=none] (3) at (1, 1.5) {};
		\node [style=none] (4) at (2, 0.25) {};
		\node [style=none] (5) at (1, -1) {};
		\node [style=none] (6) at (-1.5, 1.5) {};
		\node [style=none] (7) at (-1.5, -1) {};
	\end{pgfonlayer}
	\begin{pgfonlayer}{edgelayer}
		\draw (6.center) to (7.center);
		\draw (3.center) to (4.center);
		\draw (4.center) to (5.center);
	\end{pgfonlayer}
\end{tikzpicture}}+
\frac{1}{\sqrt{2}}
\resizebox{2cm}{!}{
\begin{tikzpicture}[baseline={([yshift=-.5ex]current bounding box.center)}]
	\begin{pgfonlayer}{nodelayer}
		\node [style=white dot] (0) at (-0.25, 1) {$H$};
		\node [style=white dot] (1) at (-0.75, -0.5) {$Z$};
		\node [style=white dot] (2) at (0.75, 0) {$H$};
		\node [style=none] (3) at (1, 1.5) {};
		\node [style=none] (4) at (2, 0.25) {};
		\node [style=none] (5) at (1, -1) {};
		\node [style=none] (6) at (-1.5, 1.5) {};
		\node [style=none] (7) at (-1.5, -1) {};
	\end{pgfonlayer}
	\begin{pgfonlayer}{edgelayer}
		\draw (6.center) to (7.center);
		\draw (3.center) to (4.center);
		\draw (4.center) to (5.center);
	\end{pgfonlayer}
\end{tikzpicture}}\\
\resizebox{2cm}{!}{
\begin{tikzpicture}[baseline={([yshift=-.5ex]current bounding box.center)}]
	\begin{pgfonlayer}{nodelayer}
		\node [style=white dot] (0) at (-0.25, 1) {$H$};
		\node [style=white dot] (1) at (-0.75, -0.5) {$H$};
		\node [style=white dot] (2) at (0.75, 0) {$H$};
		\node [style=none] (3) at (1, 1.5) {};
		\node [style=none] (4) at (2, 0.25) {};
		\node [style=none] (5) at (1, -1) {};
		\node [style=none] (6) at (-1.5, 1.5) {};
		\node [style=none] (7) at (-1.5, -1) {};
	\end{pgfonlayer}
	\begin{pgfonlayer}{edgelayer}
		\draw (6.center) to (7.center);
		\draw (3.center) to (4.center);
		\draw (4.center) to (5.center);
	\end{pgfonlayer}
\end{tikzpicture}}&=
\frac{1-i}{\sqrt{2}}
\resizebox{2cm}{!}{
\begin{tikzpicture}[baseline={([yshift=-.5ex]current bounding box.center)}]
	\begin{pgfonlayer}{nodelayer}
		\node [style=white dot] (0) at (-0.25, 1) {$I$};
		\node [style=white dot] (1) at (-0.75, -0.5) {$H$};
		\node [style=white dot] (2) at (0.75, 0) {$H$};
		\node [style=none] (3) at (1, 1.5) {};
		\node [style=none] (4) at (2, 0.25) {};
		\node [style=none] (5) at (1, -1) {};
		\node [style=none] (6) at (-1.5, 1.5) {};
		\node [style=none] (7) at (-1.5, -1) {};
	\end{pgfonlayer}
	\begin{pgfonlayer}{edgelayer}
		\draw (6.center) to (7.center);
		\draw (3.center) to (4.center);
		\draw (4.center) to (5.center);
		\draw (1) to (0);
		\draw (0) to (2);
	\end{pgfonlayer}
\end{tikzpicture}}+
\frac{1+i}{\sqrt{2}}
\resizebox{2cm}{!}{
\begin{tikzpicture}[baseline={([yshift=-.5ex]current bounding box.center)}]
	\begin{pgfonlayer}{nodelayer}
		\node [style=white dot] (0) at (-0.25, 1) {$S$};
		\node [style=white dot] (1) at (-0.75, -0.5) {$H$};
		\node [style=white dot] (2) at (0.75, 0) {$H$};
		\node [style=none] (3) at (1, 1.5) {};
		\node [style=none] (4) at (2, 0.25) {};
		\node [style=none] (5) at (1, -1) {};
		\node [style=none] (6) at (-1.5, 1.5) {};
		\node [style=none] (7) at (-1.5, -1) {};
	\end{pgfonlayer}
	\begin{pgfonlayer}{edgelayer}
		\draw (6.center) to (7.center);
		\draw (3.center) to (4.center);
		\draw (4.center) to (5.center);
		\draw (1) to (0);
		\draw (0) to (2);
	\end{pgfonlayer}
\end{tikzpicture}}\\
\resizebox{2cm}{!}{
\begin{tikzpicture}[baseline={([yshift=-.5ex]current bounding box.center)}]
	\begin{pgfonlayer}{nodelayer}
		\node [style=white dot] (0) at (-0.25, 1) {$H$};
		\node [style=white dot] (1) at (-0.75, -0.5) {$H$};
		\node [style=white dot] (2) at (0.75, 0) {$H$};
		\node [style=none] (3) at (1, 1.5) {};
		\node [style=none] (4) at (2, 0.25) {};
		\node [style=none] (5) at (1, -1) {};
		\node [style=none] (6) at (-1.5, 1.5) {};
		\node [style=none] (7) at (-1.5, -1) {};
	\end{pgfonlayer}
	\begin{pgfonlayer}{edgelayer}
		\draw (6.center) to (7.center);
		\draw (3.center) to (4.center);
		\draw (4.center) to (5.center);
	\end{pgfonlayer}
\end{tikzpicture}}&=
\resizebox{2cm}{!}{
\begin{tikzpicture}[baseline={([yshift=-.5ex]current bounding box.center)}]
	\begin{pgfonlayer}{nodelayer}
		\node [style=white dot] (0) at (-0.25, 1) {$I$};
		\node [style=white dot] (1) at (-0.75, -0.5) {$H$};
		\node [style=white dot] (2) at (0.75, 0) {$I$};
		\node [style=none] (3) at (1, 1.5) {};
		\node [style=none] (4) at (2, 0.25) {};
		\node [style=none] (5) at (1, -1) {};
		\node [style=none] (6) at (-1.5, 1.5) {};
		\node [style=none] (7) at (-1.5, -1) {};
	\end{pgfonlayer}
	\begin{pgfonlayer}{edgelayer}
		\draw (6.center) to (7.center);
		\draw (3.center) to (4.center);
		\draw (4.center) to (5.center);
		\draw (1) to (0);
		\draw (1) to (2);
	\end{pgfonlayer}
\end{tikzpicture}}+
\resizebox{2cm}{!}{
\begin{tikzpicture}[baseline={([yshift=-.5ex]current bounding box.center)}]
	\begin{pgfonlayer}{nodelayer}
		\node [style=white dot] (0) at (-0.25, 1) {$Z$};
		\node [style=white dot] (1) at (-0.75, -0.5) {$H$};
		\node [style=white dot] (2) at (0.75, 0) {$Z$};
		\node [style=none] (3) at (1, 1.5) {};
		\node [style=none] (4) at (2, 0.25) {};
		\node [style=none] (5) at (1, -1) {};
		\node [style=none] (6) at (-1.5, 1.5) {};
		\node [style=none] (7) at (-1.5, -1) {};
	\end{pgfonlayer}
	\begin{pgfonlayer}{edgelayer}
		\draw (6.center) to (7.center);
		\draw (3.center) to (4.center);
		\draw (4.center) to (5.center);
		\draw (1) to (0);
		\draw (1) to (2);
		\draw (0) to (2);
	\end{pgfonlayer}
\end{tikzpicture}}
\end{align}
\end{example}

We take a closer look at Cases \ref{case2} and \ref{case3} of Theorem \ref{generalmerge} by considering inner products, revealing the symmetries in non-Pauli-related triplets of linearly dependent stabilizer states.
\begin{theorem}
If two stabilizer states $\ket{\psi_1}$ and $\ket{\psi_2}$ satisfy $\bra{\psi_1}\psi_2\rangle\in \{\frac{i-1}{2},-\frac{1}{2}\}$, then $\ket{\psi_3}$, defined as $\ket{\psi_3}\equiv -(\ket{\psi_1}+\ket{\psi_2})$, is a stabilizer state and satisfies $\bra{\psi_2}\psi_3\rangle=\bra{\psi_3}\psi_1\rangle=\bra{\psi_1}\psi_2\rangle$.
\end{theorem}
\begin{proof}
Let $\ket{\psi_1}=C\ket{0^n}$ and $\ket{\psi}\equiv C^{\dagger}\ket{\psi_2}$ for some Clifford operator $C$. 
If $\bra{0^n}\psi \rangle =\frac{i-1}{2}$, then $\ket{\psi}$ is of the form $\frac{i-1}{2}\ket{0^n}+i^k\frac{i-1}{2}\ket{s}$ for some non-zero computational basis state $\ket{s}$ and integer $k$. Then, $\ket{\psi_3}$, which is equal to $C(-\frac{i+1}{2}\ket{0^n}-i^k\frac{i-1}{2}\ket{s})$, is a stabilizer state and 
satisfies $\bra{\psi_2}\psi_3\rangle=\bra{\psi_3}\psi_1\rangle=\frac{i-1}{2}$.
Likewise, if $\bra{0^n}\psi \rangle =-\frac{1}{2}$, then $\ket{\psi}$ is of the form $-\frac{1}{2}(\ket{0^n}+i^{k_1}\ket{s_1}+i^{k_2}\ket{s_2}+i^{k_3}\ket{s_3})$ for some distinct computational basis states $\ket{s_1},\ket{s_2},\ket{s_3}$ and some integers $k_1,k_2,k_3$, so $\ket{\psi_3}$ similarly is a stabilizer state and satisfies
$\bra{\psi_2}\psi_3\rangle=\bra{\psi_3}\psi_1\rangle=-\frac{1}{2}$.
\end{proof}
\subsection{Inner product algorithm}
We now turn our attention to computing inner products between extended graph states. 
Our inner product algorithm has cubic worst case runtime, same as the current best algorithm, based on generator matrices \cite{garcia}. 
Our algorithm is more direct in implementation due to the correspondence between an extended graph state and the gates required to produce it and is also naturally global phase sensitive.
Our algorithm uses the following observation.
\begin{proposition}
\label{inner}
Let $\ket{\psi}\equiv\bigotimes\limits_{i=1}^nc_iz_i\ket{G}$ be in reduced form, and let $A\equiv \{i|c_i=H\}$. Then
\begin{equation}\bra{0}^{\tens{}n}\ket{\psi}=\begin{cases} 
0 & \exists i\in A, z_i=Z\\
\frac{1}{\sqrt{2^{n-|A|}}} &\text{otherwise}
\end{cases}.
\end{equation}
\end{proposition}
\begin{proof}
Note that
\begin{align*}
    \bra{0}^{\tens{}n}\bigotimes\limits_{i=1}^nc_iz_i\ket{G} &=\bra{0}^{\tens{}n}\prod\limits_{p\in A}H_p(z_p)_p\ket{+}^{\tens{}n}\nonumber\\
    &=\frac{1}{\sqrt{2^{n-|A|}}}\bra{0}_A\prod\limits_{p\in A}(x_p)_p\ket{0}^{\tens{}n},
\end{align*}
where $x_p=X$ when $z_p=Z$ and $x_p=I$ when $z_p=I$. If $x_p=X$ for some $p\in A$, then $\bra{0}^{\tens{}n}\ket{\psi}=0$ and otherwise $\bra{0}^{\tens{}n}\ket{\psi}=\frac{1}{\sqrt{2^{n-|A|}}}$.
\end{proof}
We present our algorithm in the proof of the following theorem.
\begin{theorem}
\label{innerprod}
Let $\ket{\psi}\equiv \bigotimes\limits_{i=1}^nC_i\ket{G}$ and $\ket{\psi'}\equiv\bigotimes\limits_{i=1}^nC_i'\ket{G'}$ be two extended graph states. Then $\langle \psi \ket{\psi'}$ can be computed in $O(nd^2)$ time, where $d$ is the maximum degree in $G$ and $G'$ encountered during the calculation.
\end{theorem}
\begin{proof}
First we apply $C_i^{\dagger}$ to $C_i'$ for each $i$. It suffices to take the inner product of $\ket{G}$ and $\bigotimes\limits_{i=1}^nD_i\ket{G'}$ for local Clifford operators $D_i$. We do so by taking the inner product of $\ket{0}^{\tens{}n}$ and $\ket{\varphi}\equiv\bigotimes\limits_{i=1}^nH\prod\limits_{(i,j)\in E(G)}CZ_{i,j}\bigotimes\limits_{i=1}^nD_i\ket{G}$. 
We first simplify the layer of $CZ$ operators. 
\begin{definition}
\label{batch}
A \textit{star operation} on qubit $p$ is a product of $CZ$ operators, each having one of the qubits it is applied to equal $p$.
\end{definition}
For each $p\in \left[n\right]$, we apply star operations of the form $\prod\limits_{q\in N(p),q>p}CZ_{p,q}$ to $\bigotimes\limits_{i=1}^nD_i\ket{G'}$. We perform updates in $O(d^2)$ time as follows. If $D_p$ takes $Z$ to $\pm Z$ upon conjugation, then for each neighbor $q$ of $p$ we apply $CZ_{p,q}$ by the method described in Section \ref{section5}, conjugating it through $D_p$ and $D_q$ and either applying a normal $CZ$ gate to $G'$, Lemma \ref{xz}, or Lemma \ref{yz}. 
After updating $D_i$ and $G'$, $D_p$ still takes $Z$ to $\pm Z$ upon conjugation since $D_p$ is changed by a diagonal Clifford, so we can repeat the same update process for all qubits $q$. If $D_p$ takes $Z$ to $\pm X$ upon conjugation, we apply Theorem \ref{hsliding} to qubit $p$ and some neighbor $q$ of $p$, which changes $D_p$ to $D_pH$. We proceed as before because $D_p$ now takes $Z$ to $\pm Z$ upon conjugation.
If qubit $p$ does not have a neighbor, then the application of the $X$ operator to qubit $p$ does not change $\ket{G'}$, so applying $CZ_{p,q}$ becomes equivalent to applying some Pauli operator on qubit $q$, which is trivial.
If $D_p$ takes $Z$ to $\pm Y$ upon conjugation, then we apply Theorem \ref{lc} to qubit $p$, which changes $D_p$ to $D_pHS^{\dagger}H$. We proceed as in the first case because $D_p$ takes $Z$ to $\pm Z$.

Next, we append $H$ to $D_i$ for all $i$ and simplify $\bigotimes\limits_{i=1}^nD_i\ket{G'}$ to reduced form, following the algorithm described in Theorem \ref{reachability}. $\bra{\psi}\psi'\rangle$ is equal to the product of the result of Proposition \ref{inner} applied to $\ket{\varphi}$ and the global phase factors produced during the calculation.
The total runtime of the algorithm is $O(nd^2)$.
\end{proof}
\section{Conclusions}
In this paper, we explored the stabilizer formalism through the lens of the graph formalism. 
We created a canonical form for expressing extended graph states in a concise and unique way that improves upon previous reduced forms \cite{mikeike,stabgraph1,stabgraph2}.
No longer is the equivalence test \cite{stabgraph1} needed to test for equivalence, and our canonical form is easier to visualize than matrix-based representations. 
The connections between stablizer states and the properties of their corresponding graphs when expressed in canonical form can be explored in future work.
In addition, when tasked with simplifying a linear combination of stabilizer states into fewer terms, we can now convert each stabilizer state into canonical form with our efficient simplification algorithm and then combine identical representations.
With our improved inner product algorithm that runs in $O(nd^2)$, we can transform and simplify the linear combination using our characterization of linearly dependent triplets, simpler extended graph state merging formula, and extended graph state splitting formula.
This line of work can lead to an algorithm for computing more concise representations of arbitrary quantum states.

We applied our merging formulas to discover new rules that describe the action of $CZ$ gates on arbitrary extended graph states. 
Our novel transformation rules enable us to simplify GraphSim's algorithm for applying controlled-Pauli operators to graph states \cite{graphsim,kerzner} and improve runtime. We apply our transformation rules to prove that it is impossible to update extended graph states in under quadratic time in the number of qubits upon the application of a controlled-Pauli gate by toggling edges.
Therefore, in order to improve graph state simulation, we should consider algorithms that do not simply apply one gate at a time. Whenever multiple $CZ$ gates can be applied consecutively, we can apply star operations following the method described in the proof of Theorem \ref{innerprod} to spread out the $O(d^2)$ update time over multiple $CZ$ gates, improving performance.
Future work to improve graph state simulation could study the relationship between the circuit and the runtime, as well as design more efficient algorithms for simulating certain types of circuits, both those that graph state simulation is already suitable for, such as quantum error-correcting circuits, or other circuits.

We derive a simpler graph merging formula and completely characterize linearly dependent triplets of stabilizer states to extend previous work \cite{khesin ren,garcia2}.
Both our characterization in terms of extended graph states and in terms of inner products reveal much structure in the additive properties of stabilizer states that can possibly be generalized.
Future work can continue characterizing linearly dependent $n$-tuples of stabilizer states for $n\ge 4$ as well as stabilizer decompositions of magic states \cite{bss,bravyi smolin, hq, shir}, using the graph formalism.
The appendices are organized as follows. Appendix \ref{appendixa} contains proofs of Theorems \ref{hsliding} and \ref{hsplitting}, Appendix \ref{appendixb} contains a discussion of improving upper bounds on the stabilizer rank of magic states, and Appendix \ref{appendixc} contains proofs for our graph state transformation rules.

\section*{Acknowledgements}
The authors would like to thank the MIT PRIMES-USA program for making this project possible.

\section{Appendices}
\appendix
\section{Proofs of Theorems \ref{hsliding} and \ref{hsplitting}}
\label{appendixa}
Here we prove Theorem \ref{hsliding}.
\begin{proof}
Without loss of generality let $x=1$ and $y=2$, $b_1,b_2\in \{0,1\}$, and $s$ be some binary string of length $n-2$. It suffices to show \begin{multline*}
    \left(\bra{b_1b_2}\tens{}\bra{s}\right)H_xH_y\ket{G}=\\ \left(\bra{b_1b_2}\tens{}\bra{s}\right)Z_xZ_y\prod\limits_{p\in A,q\in B}CZ_{p,q}\ket{G}.
\end{multline*} 
Let $s_i$ denote the bit in $s$ corresponding to qubit $i$, where $3\le i\le n$. Let $a\equiv |\{p\in A\setminus B | s_p=1\}|$, $b\equiv |\{p\in A\cap B\setminus\{1,2\} | s_p=1\}|$, $c\equiv |\{p\in B\setminus A|s_p=1 \}|$. Then, the left hand side can be evaluated in terms of $a,b,c$ as follows: \begin{multline*}
        \left(\bra{b_1b_2}\tens{}\bra{z}\right)H_1H_2\ket{G}\\=\frac{1}{2} \sum\limits_{(j,k)\in\{0,1\}^2}(-1)^{b_1j+b_2k}\left(\bra{jk}\tens{}\bra{z}\right)\ket{G}\\
    =\frac{1}{2} \sum\limits_{(j,k)\in\{0,1\}^2}(-1)^{b_1j+b_2k+j(a+b)+k(b+c)+jk}\\ \cdot\left(\bra{00}\tens{}\bra{z}\right)\ket{G},
\end{multline*}
while for the right hand side, letting $A'\equiv A\setminus\{1,2\},B'\equiv B\setminus \{1,2\}$,\begin{multline*}
    \left(\bra{b_1b_2}\tens{}\bra{z}\right)Z_xZ_y\prod\limits_{p\in A,q\in B}\ket{G}\\ =\left(\bra{b_1b_2}\tens{}\bra{z}\right)\prod\limits_{p\in A'}CZ_{p,1}CZ_{p,2}\\ \cdot \prod\limits_{q\in B'}CZ_{1,q}CZ_{2,q}\prod\limits_{p\in A',q\in B'}CZ_{p,q}\ket{G}\\=(-1)^{b_1(a+c)+b_2(a+c)+ab+bc+ca+b}\left(\bra{b_1b_2}\tens{}\bra{z}\right)\ket{G}\\=(-1)^{\zeta}\left(\bra{00}\tens{}\bra{z}\right)\ket{G},
\end{multline*}
where $\zeta = b_1(a+c)+b_2(a+c)+ab+bc+ca+b+b_1(a+b)+b_2(b+c)+b_1 b_2$.
It suffices to verify 
\begin{multline*}
    \frac{1}{2}\sum\limits_{(j,k)\in\{0,1\}^2}(-1)^{b_1j+b_2k+j(a+b)+k(b+c)+jk}=(-1)^{\zeta}
\end{multline*}
for all $(a,b,c,b_1,b_2)\in\{0,1\}^5$.
\end{proof}
Here we prove Theorem \ref{hsplitting}.
\begin{proof}
Without loss of generality let $x=1$ and $y=2$, $b_1,b_2\in \{0,1\}$, and $s$ be some binary string of length $n-2$. It suffices to show \begin{multline}
\label{eqna5}
\left(\bra{b_1b_2}\tens{}\bra{s}\right)H_xH_y\ket{G} =\left(\bra{b_1b_2}\tens{}\bra{s}\right)Z_xZ_y\ket{G}\\+\left(\bra{b_1b_2}\tens{}\bra{s}\right)\prod\limits_{p\in A,q\in B}CZ_{p,q}\ket{G}.\end{multline}
Let $s_i$ denote the bit in $s$ corresponding to qubit $i$, where $3\le i\le n$. Let $a\equiv |\{p\in A\setminus (B\cup\{1\}) | s_p=1\}|$, $b\equiv |\{p\in A\cap B| s_p=1\}|$, and $c\equiv |\{p\in B\setminus (A\cup\{1\})|s_p=1 \}|$. Then, the left hand side can be evaluated in terms of $a,b,c$ similarly as in the proof of Theorem \ref{hsliding}:
\begin{multline*}
    \left(\bra{b_1b_2}\tens{}\bra{s}\right)H_xH_y\ket{G}= \frac{1}{2}\\ \cdot \sum\limits_{(j,k)\in\{0,1\}^2}(-1)^{b_1j+b_2k+j(a+b)+k(b+c)}\left(\bra{00}\tens{}\bra{z}\right)\ket{G}.
\end{multline*}
The first term in the right hand side of Equation \ref{eqna5} is
\begin{multline*}
    \left(\bra{b_1b_2}\tens{}\bra{s}\right)Z_xZ_y\ket{G}=(-1)^{b_1+b_2}\left(\bra{b_1b_2}\tens{}\bra{s}\right)\\ \cdot \ket{G}
    =(-1)^{b_1+b_2+b_1(a+b)+b_2(b+c)}\left(\bra{00}\tens{}\bra{s}\right)\ket{G},
\end{multline*}
while the second term in the right hand side is
\begin{multline*}
    \left(\bra{b_1b_2}\tens{}\bra{s}\right)\prod\limits_{p\in\nbhd(x)}Z_p\prod\limits_{q\in\nbhd(y)}Z_q\prod\limits_{p\in A,q\in B}CZ_{p,q}\ket{G}\\
    =(-1)^{a+c}\left(\bra{b_1b_2}\tens{}\bra{s}\right)CZ_{1,2}\\ \cdot \prod\limits_{p\in A}CZ_{p,2}\prod\limits_{q\in B}CZ_{1,q}\prod\limits_{p\in \nbhd(1),q\in \nbhd(2)}CZ_{p,q}\ket{G}\\
    =(-1)^{a+c+b_1b_2+b_1(b+c)+b_2(a+b)+ab+bc+ca+b}\\\cdot \left(\bra{b_1b_2}\tens{}\bra{s}\right)\ket{G}\\
    =(-1)^{a+b+c+b_1b_2+(b_1+b_2)(a+c)+ab+bc+ca}\\ \cdot \left(\bra{00}\tens{}\bra{s}\right)\ket{G}.
\end{multline*}
It suffices to verify
\begin{multline*}
    \frac{1}{2} \sum\limits_{(j,k)\in\{0,1\}^2}(-1)^{b_1j+b_2k+j(a+b)+k(b+c)}\\=(-1)^{b_1+b_2+b_1(a+b)+b_2(b+c)}\\ +(-1)^{a+c+b_1b_2+b_1(a+c)+b_2(a+c)+ab+bc+ca+b}
\end{multline*}
for all $(a,b,c,b_1,b_2)\in\{0,1\}^5$.
\end{proof}
\section{Stabilizer rank of magic states}
\label{appendixb}
Here we discuss our attempts at finding upper bounds on the stabilizer rank of $n$-qubit magic states, which we first define.
\begin{definition}
A \textit{n-qubit magic state} $\ket{T_n}$ is the state $\ket{T}^{\tens{}n}$, where $\ket{T}\equiv \frac{\ket{0}+e^{\frac{\pi i}{4}}\ket{1}}{\sqrt{2}}$.
\end{definition}
\begin{definition}
The $\textit{stabilizer rank}$ $\chi(\ket{\psi})$ of a state $\ket{\psi}$ is the smallest integer $\chi$ such that there exists a set of $\chi$ stabilizer states $S$ such that $\ket{\psi}\in \text{span}(S)$.
\end{definition}
The stabilizer rank of the magic state is deeply tied to the runtimes of classical simulations of quantum circuits and has been explored extensively \cite{bravyi smolin,shir,bss}.
In order to tighten the upper bounds on $\chi(\ket{T_n})$, a 
Metropolis-Hastings numerical search algorithm that applies random transformations to stabilizer states to maximize the projection of the magic state onto their span was developed \cite{bravyi smolin}.
We experimented with different variations of this method and were unable to find stabilizer decompositions for higher numbers of qubits, due to the extremely large search space.
Another method, introduced in \cite{hq}, utilizes \textit{cat states} and \textit{contractions}. It produced the best known upper bounds on $\chi(\ket{T_n})$ by enabling a $6$ qubit decomposition of the $6$ qubit magic state to be found by inspection. However, for higher qubit states, inspection cannot be used and computing stabilizer decompositions once again becomes difficult.

Therefore, we tried a new method, which was to represent the $n$ qubit magic state, $(T\ket{+})^{\tens{} n}$, as a linear combination of extended graph states written in our canonical form. By comparing pairs of extended graph states, we can see whether they can be merged together. 
We were able to use Mathematica to express the $n$ qubit magic state as a sum of $2^{\frac{n}{2}}$ stabilizer states. At this point, though, we were not able to apply any more merges. Future work could try to develop more general merging criteria and formulas involving more than three extended graph states. Then, a computer could continually transform the sum of extended graph states following some sort of heuristic that allows it to stumble upon an optimal decomposition with some luck. 
In order to find such a heuristic, it would be useful to study properties of sums of extended graph states that provide more insights into the structure of low-rank stabilizer decompositions. 
So, we converted the stabilizer decompositions found in \cite{bravyi smolin,hq} into our canonical form to see if we could glean any insights about their structure. We provide examples of stabilizer decompositions for two special cases, $n=3$ and $n=6$.
\begin{align}
    &(T\ket{+})^{\tens{} 3}=\frac{i-e^{i\frac{\pi}{4}}}{2}Z_1Z_2Z_3\ket{I_3}\nonumber\\&-\frac{i+e^{\frac{\pi i}{4}}}{2}Z_1Z_2Z_3\ket{K_3}+\frac{1+e^{i\frac{\pi}{4}}}{2}H_1H_2S_3\ket{S_{3,3}}\nonumber\\
    &(T\ket{+})^{\tens{}6}= -\frac{\sqrt{i}}{2\sqrt{2}}H_6(HZ)^{\tens{}6}\ket{S_{6,6}}\nonumber\\&+\frac{1}{2\sqrt{2}}H^{\tens{}6}H_6S_6Z_6\ket{S_{6,6}}-\frac{i}{2}H_1S_1S^{\tens{}6}\ket{S_{6,1}}\nonumber\\&+\frac{i\sqrt{i}}{2}H_1\ket{S_{6,1}}-\frac{1}{2}H_1S_1Z_1(SZ)^{\tens{}6}\ket{K_6}\nonumber\\&-\frac{\sqrt{i}}{2}H_1Z_1Z^{\tens{}6}\ket{K_6},
\end{align}
where $S_{n,i}$ is the star graph on $n$ vertices with central vertex $i$. Even for the $3$ qubit case, there are multiple ways to decompose the magic state into $3$ stabilizer states, yet in all of the ways, there seems to be an empty graph, a complete graph, and a star graph. Future work can completely characterize the $3$ qubit case and move on to higher cases.
\section{Proofs of CZ transformation rules}
\label{appendixc}
Here we prove the expressions for $\ket{\psi_{PQ}}$ in Table \ref{table1}. $\ket{\psi_{ZZ}}$ is trivial.
\begin{lemma}[Elliot et al.]
\label{xz}
    \begin{equation}
        \ket{\psi_{ZX}}=\prod\limits_{x\in N_2}CZ_{1,x}\ket{G}.
    \end{equation}
\end{lemma}
\begin{proof}
 Let $G'$ be the graph formed from $G$ where all the edges incident to $1$ are removed. Suppose $1$ is connected to $2$ in $G$. Applying Theorem \ref{khesin}, we have
\begin{multline*}
\frac{1}{2}\left(\underbrace{(I+Z)_1\ket{G}}_{A=M_1,B=\{1\}}+X_2\underbrace{(I-Z)_1\ket{G}}_{A=M_1,B=\{1\}}\right)\\= \frac{1}{\sqrt{2}}H_1\ket{G'}+X_2\frac{1}{\sqrt{2}}H_1Z_1\left(\prod\limits_{x\in N_1}Z_x\right)\ket{G'}\\
=\frac{1}{\sqrt{2}}H_1\underbrace{\left(I-Z_1\prod\limits_{x\in N_1}Z_x\prod\limits_{x\in N_2\setminus\{1\}}Z_x\right)\ket{G'}}_{A=\{1\},B=N_1\triangle N_2}\\
=\prod_{x\in N_1\triangle N_2}CZ_{1,x}\ket{G'}
=\prod_{x\in N_2}CZ_{1,x}\ket{G}
\end{multline*}
Now suppose $1$ is not connected to $2$ in $G$. We have
\begin{multline*}
\frac{1}{2}\left(\underbrace{(I+Z)_1\ket{G}}_{A=M_1,B=\{1\}}+X_2\underbrace{(I-Z)_1\ket{G}}_{A=M_1,B=\{1\}}\right)\\= \frac{1}{\sqrt{2}}H_1\ket{G'}+X_2\frac{1}{\sqrt{2}}H_1Z_1\left(\prod\limits_{x\in N_1}Z_x\right)\ket{G'}\\
=\frac{1}{\sqrt{2}}H_1\underbrace{\left(I+Z_1\prod\limits_{x\in N_1}Z_x\prod\limits_{x\in N_2}Z_x\right)\ket{G'}}_{A=\{1\},B=M_1\triangle N_2}\\
=Z_1\prod_{x\in M_1\triangle N_2}CZ_{1,x}\ket{G'}=\prod_{x\in N_2}CZ_{1,x}\ket{G}
\end{multline*}
\end{proof}
\begin{lemma}
\label{yz}
\begin{equation}
    \ket{\psi_{YZ}}=S_2Z_2\prod\limits_{x\in M_1}CZ_{2,x}\ket{G}.
\end{equation}
\end{lemma}
\begin{proof}
Let $G'$ refer to the graph resulting from toggling all the edges between vertices in the set $M_1$ in $G$. Applying Theorem \ref{graphmerging}, we have
\begin{multline*}
\frac{1}{2}\left(\underbrace{(I-iZ_1X_1)\ket{G}}_{A=M_1}+Z_2\underbrace{(I+iZ_1X_1)\ket{G}}_{A=M_1}\right)\\=\frac{1-i}{2}S_1\prod\limits_{x\in  N_1}S_x\ket{G'}+Z_2\frac{1+i}{2}S_1^3\prod\limits_{x\in N_1}S_x^3\ket{G'}\\
=\frac{1-i}{2}S_1\prod\limits_{x\in N_1}S_x\underbrace{\left(I+iZ_2\prod\limits_{x\in N_1}Z_x\right)\ket{G'}}_{A=M_1 \triangle\{2\}}\\
=S_1\prod\limits_{x\in N_1}S_x\prod\limits_{x\in M_1\triangle \{2\}}S_x^3\prod\limits_{x\in M_1\setminus \{2\}}CZ_{2,x}\ket{G}\\
=S_2Z_2\prod\limits_{x\in M_1}CZ_{2,x}\ket{G}
\end{multline*}where the final step can be seen from case work on whether $2\in N_1$.
\end{proof}
\begin{lemma}[Elliot et al.]
\label{xx}
If $(1,2)\in E(G)$,
\begin{equation}
    \ket{\psi_{XX}}=H_1H_2CZ_{1,2}\prod\limits_{x\in M_1,y\in M_2}CZ_{x,y}\ket{G}.
\end{equation}
Otherwise,
\begin{equation}
    \ket{\psi_{XX}}=\prod\limits_{x\in N_1,y\in N_2}CZ_{x,y}\ket{G}.
\end{equation}
\end{lemma}
\begin{proof}
If $1$ and $2$ are connected in $G$, we apply Theorem \ref{hsliding}
\begin{align*}
    \ket{\psi_{XX}}&=H_1H_2CZ_{1,2}H_1H_2\ket{G}\nonumber\\
    &=H_1H_2CZ_{1,2}\prod\limits_{x\in M_1,y\in M_2}CZ_{x,y}\ket{G}
\end{align*}
If $1$ and $2$ are not connected we follow the proof given in \cite{stabgraph1}.
\end{proof}
\begin{lemma}
\label{xy}
If $(1,2)\in E(G)$,
\begin{equation}
    \ket{\psi_{YX}}=\frac{1-i}{\sqrt{2}}\left(\prod\limits_{x\in M_1}S_x\right) H_1 \prod\limits_{x\in M_1\triangle M_2}CZ_{1,x}\ket{L_1(G)}.
\end{equation}Otherwise,
\begin{equation}
    \ket{\psi_{YX}}=\prod\limits_{x\in M_1\triangle N_2}Z_x\prod\limits_{x,y\in M_1\triangle N_2}CS_{x,y}\prod\limits_{x,y\in M_1}CS_{x,y}\ket{G}.
\end{equation}
\end{lemma}
\begin{proof}
Let $G'$ be $G$ with all edges between vertices in $M_1$ toggled. If vertices $1$ and $2$ are connected in $G$,
\begin{multline*}
\frac{1}{2}\left( \underbrace{(I-iZ_1X_1)\ket{G}}_{A=M_1}+X_2\underbrace{(I+iZ_1X_1)\ket{G}}_{A=M_1} \right)\\=\frac{1-i}{2}\prod\limits_{x\in M_1}S_x\ket{G'}+Y_2\frac{1+i}{2}\prod\limits_{x\in M_1}S_x^3\ket{G'}\\
=\frac{1-i}{2}\prod\limits_{x\in M_1}S_x\ket{G_1}+(1+i)\prod\limits_{x\in M_1}S_x^3 Y_2\ket{G'}\\
=\frac{1-i}{2}\prod\limits_{x\in M_1}S_x\underbrace{\left(I+Z_2\prod\limits_{x\in M_1}Z_x\prod\limits_{x\in M_2\triangle M_1}Z_x\right)\ket{G'}}_{A=\{1\},B=N_2}
\end{multline*}
\begin{multline*}
=\frac{1-i}{\sqrt{2}}\left(\prod\limits_{x\in M_1}S_x\right)H_1\prod\limits_{x\in N_2}CZ_{1,x}\ket{G'}\\
=\frac{1-i}{\sqrt{2}}\left(\prod\limits_{x\in M_1}S_x\right)H_1\prod\limits_{x\in N_1\cup N_2}CZ_{1,x}\ket{L_1(G)}
\end{multline*}If vertices $1$ and $2$ are not connected in $G$, \begin{multline*}
    \frac{1}{2}\left( \underbrace{(I-iZ_1X_1)\ket{G}}_{A=M_1}+X_2\underbrace{(I+iZ_1X_1)\ket{G}}_{A=M_1} \right)\\=\frac{1-i}{2}\prod\limits_{x\in M_1}S_x\ket{G'}+X_2\frac{1+i}{2}\prod\limits_{x\in M_1}S_x^3\ket{G'}\\
    =\frac{1-i}{2}\prod\limits_{x\in M_1}S_{x}\underbrace{\left(I+iZ_2\prod\limits_{x\in M_1}Z_x\prod\limits_{x\in M_2}Z_x\right)\ket{G'}}_{A=M_1\triangle N_2}\\
    =\prod\limits_{x\in N_1}S_x \prod\limits_{x\in M_1\triangle N_2}Z_x \prod\limits_{x,y\in M_1\triangle N_2}CS_{x,y}\ket{G'}\\
    =\prod\limits_{x\in M_1\triangle N_2}Z_x\prod\limits_{x,y\in M_1\triangle N_2}CS_{x,y}\prod\limits_{x,y\in M_1}CS_{x,y}\ket{G}
\end{multline*}
\end{proof}
\begin{lemma}
\label{yy}
If $(1,2)\in E(G)$,
\begin{multline}
    \ket{\psi_{YY}}=-i\prod\limits_{x,y\in M_1}CS_{x,y}\prod\limits_{x,y\in M_2}CS_{x,y}\ket{G}.
\end{multline}Otherwise,
\begin{multline}
    \ket{\psi_{YY}}=\frac{1-i}{\sqrt{2}}\prod\limits_{x\in M_1}S_x H_1 \prod\limits_{x\in M_2}CZ_{1,x}\ket{L_1(G)}.
\end{multline}
\end{lemma}
\begin{proof}
Let $G'$ be $G$ with all edges between vertices in $M_1$ toggled. If vertices $1$ and $2$ are connected in $G$,
\begin{multline*}
\frac{1}{2}\left( \underbrace{(I-iZ_1X_1)\ket{G}}_{A=M_1}+Y_2\underbrace{(I+iZ_1X_1)\ket{G}}_{A=M_1} \right)\\=\frac{1-i}{2}\prod\limits_{x\in M_1}S_x\ket{G'}+Y_2\frac{1+i}{2}\prod\limits_{x\in M_1}S_x^3\ket{G'}\\
=\frac{1-i}{2}\prod\limits_{x\in M_1}S_x\ket{G'}-(1+i)\prod\limits_{x\in M_1}S_x^3 X_2\ket{G'}\\
=\frac{1-i}{2}\prod\limits_{x\in M_1}S_x\left(I-i\prod\limits_{x\in M_1}Z_x\prod\limits_{x\in M_2\triangle M_1}Z_x\right)\ket{G'}\\
=\frac{1-i}{2}\prod\limits_{x\in M_1}S_x\underbrace{\left(I-i\prod\limits_{x\in M_2}Z_x\right)\ket{G'}}_{A=M_2}\\
=-i\prod\limits_{x\in M_1}S_x\prod\limits_{x,y\in M_2}CS_{x,y}\ket{G'}\\
=-i\prod\limits_{x,y\in M_1}CS_{x,y}\prod\limits_{x,y\in M_2}CS_{x,y}\ket{G}
\end{multline*}If vertices $1$ and $2$ are not connected in $G$, \begin{multline*}
    \frac{1}{2}\left( \underbrace{(I-iZ_1X_1)\ket{G}}_{A=M_1}+Y_2\underbrace{(I+iZ_1X_1)\ket{G}}_{A=M_1} \right)\\=\frac{1-i}{2}\prod\limits_{x\in M_1}S_x\ket{G'}+Y_2\frac{1+i}{2}\prod\limits_{x\in M_1}S_x^3\ket{G'}\\
    =\frac{1-i}{2}\prod\limits_{x\in M_1}S_{x}\underbrace{\left(I+\prod\limits_{x\in M_1}Z_x\prod\limits_{x\in M_2}Z_x\right)\ket{G'}}_{A=\{1\},B=M_1\triangle M_2}\\
    =\frac{1-i}{\sqrt{2}}\prod\limits_{x\in M_1}S_x H_1Z_1\prod_{x\in M_1\triangle M_2}CZ_{1,x}\ket{G'}\\
    =\frac{1-i}{\sqrt{2}}\prod\limits_{x\in M_1}S_x H_1 \prod_{x\in M_2}CZ_{1,x}\ket{L_1(G)}
\end{multline*}
\end{proof}

\end{document}